\documentclass[11pt]{article}

\usepackage{amsthm}
\usepackage{mathtools}
\usepackage{amsfonts}
\usepackage{amssymb}
\usepackage{bbm}

\usepackage{booktabs}
\usepackage[a4paper, total={8.5in, 11in},margin=1in]{geometry}
 
\usepackage{comment}
\usepackage[labelfont=bf]{caption}

\usepackage{booktabs}
\usepackage{siunitx}

\usepackage{dsfont}

\usepackage[ruled,vlined]{algorithm2e}

\usepackage{xcolor}

\usepackage[colorlinks=true,linkcolor=purple,citecolor=blue]{hyperref}
\usepackage{cleveref}
\DeclarePairedDelimiterX{\infdivx}[2]{(}{)}{%
  #1\;\delimsize\|\;#2%
}

\newcommand\ddfrac[2]{\frac{\displaystyle #1}{\displaystyle #2}}

\newtheorem{theorem}{Theorem}[section]
\newtheorem{definition}[theorem]{Definition}
\newtheorem{proposition}[theorem]{Proposition}
\newtheorem{corollary}[theorem]{Corollary}
\newtheorem{lemma}[theorem]{Lemma}
\newtheorem{fact}[theorem]{Fact}

\SetAlgorithmName{Mechanism}{mechanism}{List of Mechanisms}

\LinesNumbered

\DeclareMathOperator*{\E}{\mathbb{E}}
\DeclareMathOperator*{\V}{\mathbb{V}}
\DeclareMathOperator*{\argmax}{arg\,max}
\DeclareMathOperator{\dis}{d}

\DeclareMathOperator{\pr}{\mathbb{P}}
\let\Pr\relax
\DeclareMathOperator{\Pr}{\mathbb{P}}

\DeclareMathOperator{\SC}{\textsc{SC}}
\DeclareMathOperator{\SW}{\textsc{SW}}
\DeclareMathOperator{\med}{\textsf{median}}
\DeclareMathOperator{\algomedian}{\textsc{ApproxMedianViaSampling}}
\DeclareMathOperator{\algoplurality}{\textsc{ApproxPluralityViaSampling}}
\DeclareMathOperator{\algopercentile}{\textsc{ApproxPercentileViaSampling}}
\DeclareMathOperator{\algoascending}{\textsc{AscendingAuctionViaSampling}}
\DeclareMathOperator{\algomultiunit}{\textsc{MultiUnitAuctionViaSampling}}

\author{
  Ioannis Anagnostides\\[-2mm]
  Carnegie Mellon University\footnote{Work performed while at the National Technical University of Athens.}\\[-2mm]
  \texttt{ioannis.anagnostides@gmail.com}
  \and
  Dimitris Fotakis\\[-2mm]
  National Technical University of Athens\\[-2mm]
  \texttt{fotakis@cs.ntua.gr}
  \and 
  Panagiotis Patsilinakos\\[-2mm]
  National Technical University of Athens\\[-2mm]
  \texttt{patsilinak@corelab.ntua.gr}
}

\date{}                     

\title{Sampling and Optimal Preference Elicitation in Simple Mechanisms\footnote{A preliminary version of this work appeared at the Symposium on Algorithmic Game
Theory (SAGT) 2020~\cite{Anagnostides20:Asymptotically}.}}

\begin{document}

\maketitle
\pagenumbering{gobble}

\begin{abstract}

In this work we are concerned with the design of efficient mechanisms while eliciting limited information from the agents. First, we study the performance of sampling approximations in facility location games. Our key result is to show that for any $\epsilon > 0$, a sample of size $c(\epsilon) = \Theta(1/\epsilon^2)$ yields in expectation a $1 + \epsilon$ approximation with respect to the optimal social cost of the generalized median mechanism on the metric space $(\mathbb{R}^d, \| \cdot \|_1)$, while the number of agents $n \to \infty$. Moreover, we study a series of exemplar environments from auction theory through a communication complexity framework, measuring the expected number of bits elicited from the agents; we posit that any valuation can be expressed with $k$ bits, and we mainly assume that $k$ is independent of the number of agents $n$. In this context, we show that Vickrey's rule can be implemented with an expected communication of $1 + \epsilon$ bits from an average bidder, for any $\epsilon > 0$, asymptotically matching the trivial lower bound. As a corollary, we provide a compelling method to increment the price in an English auction. We also leverage our single-item format with an efficient encoding scheme to prove that the same communication bound can be recovered in the domain of additive valuations through simultaneous ascending auctions, assuming that the number of items is a constant. Finally, we propose an ascending-type multi-unit auction under unit demand bidders; our mechanism announces at every round two separate prices and is based on a sampling algorithm that performs approximate selection with limited communication, leading again to asymptotically optimal communication. Our results do not require any prior knowledge on the agents' valuations, and mainly follow from natural sampling techniques.

\end{abstract}

\clearpage

\tableofcontents

\clearpage

\pagenumbering{arabic}

\section{Introduction}

Efficient \emph{preference elicitation} has been a central theme and a major challenge from the inception of \emph{mechanism design}, with a myriad of applications in multi-agent environments and modern artificial intelligence systems. Indeed, requesting from every agent to communicate \emph{all} of her preferences is considered widely impractical, and a substantial body of work has explored alternative approaches to truncate the elicited information, sequentially asking a series of natural queries in order to elicit only the relevant parts of the information.

This emphasis has been strongly motivated for a number of reasons. First, it has been acknowledged by behavioral economists that soliciting information induces a high \emph{cognitive cost} \cite{10.1007/3-540-48835-9_2,561616}, and agents may even be reluctant to reveal their complete (private) valuation; as pertinent evidence, the superiority of \emph{indirect} mechanisms is often cited \cite{Ausubel06thelovely}. In fact, in certain domains with severe communication constraints  \cite{AgorinoCommRestrictionsArticle,10.1086/676931} a \emph{direct revelation} mechanism---in which every agent has to disclose her entire preferences---is considered even infeasible. Indeed, communication is typically recognized as the main bottleneck in distributed environments \cite{10.5555/2821576}. As explained by Blumrosen and Feldman \cite{DBLP:conf/sigecom/BlumrosenF06}, agents typically operate with a truncated action space due to technical, behavioral or regulatory reasons. Finally, a mechanism with efficient preference elicitation would provide stronger information-privacy guarantees \cite{10.2307/41409970}.

Specifically, the general question of whether a social choice function can be accurately approximated by less than the full set of agents constitutes one of the main themes in computational social choice. Perhaps the most standard approach to truncate the elicited information consists of \emph{sampling}. More precisely, given that in many real-world applications it might be infeasible to gather preferences from all the agents, the designer performs preference aggregation by randomly selecting a small subset of the entire population \cite{CANETTI199517,Dhamal2013ScalablePA}. This approach is particularly familiar in the context of \emph{voting}, where the goal is typically to \emph{predict} the outcome of the full information mechanism without actually holding the election for the entire population \cite{10.5555/2772879.2773334}; concrete examples of the aforementioned scenarios include election polls, exit polls, as well as online surveys.

In the first part of our work, we follow this long line of research in computational social choice. Specifically, we analyze the \emph{sample complexity} of the celebrated \emph{median} mechanism in the context of \emph{facility location games}, where every agent is associated with a point---corresponding to her preferred location---on some underlying metric space. 
The median mechanism is of particular importance in social choice. Indeed, the celebrated \emph{Gibbard-Satterthwaite} impossibility theorem \cite{10.2307/1914083,RePEc:eee:jetheo:v:10:y:1975:i:2:p:187-217} states that for any \emph{onto}---for every alternative there exists a voting profile that would make that alternative prevail---and non-dictatorial voting rule, there are instances for which an agent is better off casting a vote that does not correspond to her true preferences---i.e. the rule is not \emph{strategy-proof}. Importantly, this impediment disappears when the agents' preferences are restricted. Arguably the most well-known such restriction is that of the \emph{single-peaked preferences},\footnote{More precisely, suppose that the alternatives are ordered on a line, representing their \emph{positions}; as argued in \emph{spatial voting theory} \cite{RePEc:cup:cbooks:9780521275156,arrow_1990}, it is often reasonable to assume that the alternatives can be represented as points on a line (e.g., in political elections a candidate's position may indicate whether she is a ``left-wing'' or a ``right-wing'' candidate). An agent's preferences are \emph{single-peaked} if she prefers alternatives which are closer to her peak. We remark that in single-peaked domains it is known that there can be no \emph{Condorcet cycles}.} introduced by Black \cite{10.2307/1825026}, for which Moulin's \cite{10.2307/30023824} median mechanism is indeed strategy-proof. 

In the second part of our work, we endeavor to design auctions with minimal \emph{communication complexity}. As a concrete motivating example, we consider a \emph{single-item} auction, and we assume that every valuation can be expressed with $k$ bits. An important observation is that the most dominant formats are very inefficient from a communication standpoint. Indeed, \emph{Vickrey's}---or \emph{sealed-bid}---auction~\cite{Vickrey61:Counter} is a direct revelation mechanism, eliciting the entire private information from the $n$ agents, leading to a communication complexity of $n \cdot k$ bits; in fact, it should be noted that although Vickrey's auction possesses many theoretically appealing properties, its ascending counterpart exhibits superior performance in practice \cite{10.1257/0002828043052330,Ausubel06thelovely,10.2307/1913557,10.2307/2234706}, for reasons that mostly related to the simplicity, the transparency, as well as the privacy guarantees of the latter format. Unfortunately, implementing Vickrey's rule through an \emph{English} auction requires---in the worst case---exponential communication of the order $n \cdot 2^k$, as the auctioneer has to cover the entire valuation space. Thus, the elicitation pattern in an English auction is widely inefficient, and the lack of prior knowledge on the agents' valuations would dramatically undermine its performance. 


\subsection{Our Contributions}

Our work provides several new insights on sampling and preference elicitation for a series of exemplar environments from mechanism design. 

\paragraph{Facility Location Games.} First, we turn our attention to facility location games; specifically, we consider Moulin's \emph{median} mechanism. We observe that unlike the median as a function, the \emph{social cost} of the median exhibits a \emph{sensitivity} property. Subsequently, we show that for any $\epsilon > 0$, a random sample of size $c = \Theta(1/\epsilon^2)$ suffices to recover a $1 + \epsilon$ approximation w.r.t. the optimal social cost of the full information median in the metric space $(\mathbb{R}, |\cdot|)$, while the number of agents $n \to \infty$; this guarantee is established both in terms of expectation (\Cref{theorem:1-median_expectation}), and with high probability (\Cref{corollary:1-median_whp}). Consequently, it is possible to obtain a near-optimal approximation with an arbitrarily small fraction of the total information. Our analysis is quite robust, implying directly the same characterization for the median on \emph{simple} and \emph{open} curves. Next, we extend this result for the \emph{generalized median} in high-dimensional metric spaces $(\mathbb{R}^d, \| \cdot \|_1)$ in \Cref{theorem:generalized-median_expectation}. In contrast, the sensitivity property of the median does not extend on trees, as implied by \Cref{proposition:median-trees}. Finally, for completeness, we show that sampling cannot provide meaningful guarantees w.r.t. the expected social cost when allocating at least $2$ facilities on the line through the \emph{percentile} mechanism (\Cref{proposition:percentile}).

These results constitute natural continuation on efficient preference elicitation and sampling in social choice \cite{Dhamal2013ScalablePA,10.5555/2772879.2773334,CANETTI199517,10.5555/1641503.1641507}, and supplement the work of Feldman et al. \cite{feldman2015voting}; yet, to the best of our knowledge, we are the first to investigate the performance of sampling in facility location games. We stress that our guarantees do not require any prior knowledge or any discretization assumptions. Moreover, the sensitivity property of the median's social cost could be of independent interest, as it can be potentilly employed to design \emph{differentially private} \cite{DBLP:conf/tamc/Dwork08} and \emph{noise-tolerant} implementations of the median mechanism. From a technical standpoint, one of our key contributions is a novel asymptotic characterization of the rank of the sample's median (\Cref{theorem:convergence}), discussed further in the subsection below.

\paragraph{Auctions.} Next, in \Cref{section:auctions}, we espouse a communication complexity framework in order to design a series of auctions with optimal preference elicitation. In particular, we measure the number of bits elicited from the agents in expectation, endeavoring to minimize it. We mainly make the natural assumption that the number of bits that can represent any valuation---expressed with parameter $k$---is \emph{independent} on the number of agents $n$; thus, we focus on the communication complexity while $n$ asymptotically grows. In this context, we show that we can asymptotically match the trivial lower bound of $1$ bit per bidder for a series of fundamental settings, without possessing any prior knowledge.

\begin{itemize}
    \item First, we propose an ascending auction in which the ascend of the price is calibrated \emph{adaptively} through a sampling mechanism. Thus, in \Cref{corollary:optimal-single_item} we establish that we can implement Vickrey's rule with only $1 + \epsilon$ bits from an average bidder, for any $\epsilon > 0$;
    \item Moreover, we consider the design of a multi-item auction with $m$ items and bidders with additive valuations. Our main contribution is to design an efficient \emph{encoding scheme} that substantially truncates the communication in a \emph{simultaneous} implementation, asymptotically recovering the same optimal bound whenever the number of items $m$ is a small \emph{constant} (\Cref{theorem:simultaneous});
    \item Finally, we develop a novel ascending-type multi-unit auction in the domain of unit-demand bidders. Our proposed auction announces in every round two separate prices---based on a natural sampling algorithm (see \Cref{theorem:estimation}), leading again to the optimal communication bound (\Cref{theorem:multi-unit-optimal}).
\end{itemize}

Our work supplements prior work \cite{REICHELSTEIN198432,FADEL20091895} by showing that the incentive compatibility constraint does not augment the communication requirements of the interaction process for a series of fundamental settings. We also corroborate on one of the main observations of Blumrosen et al. \cite{DBLP:journals/corr/abs-1110-2733}: \emph{asymmetry helps}---in deriving tight communication bounds. Indeed, in our mechanisms the information elicitation is substantially asymmetrical. Finally, we believe that our results could have practical significance due to their simplicity and their communication efficiency.

\subsection{Overview of Techniques}

\paragraph{Facility Location Games.} In \Cref{section:median}, our first key observation in \Cref{lemma:sensitivity} is that the social cost of the median admits a \emph{sensitivity} property. Thus, it is possible to obtain a near-optimal approximation w.r.t. the optimal social cost even if the allocated facility is very far from the median. The sensitivity of the social cost of the median essentially reduces a near-optimal approximation to the concentration of the rank of the sample's median. Based on this insight, we prove in \Cref{theorem:convergence} that when the participation is large (i.e. $n \to \infty$) the distribution of the rank of the sample's median converges to a continuous \emph{transformed beta} distribution. This result should not be entirely surprising given that---as is folklore in statistics---the order statistics of the uniform distribution on the unit interval have marginal distributions belonging to the beta distribution family (e.g., see \cite{David2011}). Having made these connections, the rest of our guarantees in \Cref{section:median} for more general metric spaces follow rather directly.

\paragraph{Auctions.} With regard to our results in \Cref{section:auctions}, we commence our overview with the single-item auction, and we then describe our approach for several extensions. In particular, for our ascending auction we consider as a black-box an algorithm that implements a second-price auction; then, at every round of the auction we \emph{simulate} a sub-auction on a random sample of agents, and we broadcast the ``market-clearing price'' in the sub-auction as the price of the round. From a technical standpoint, we show that as the size of the sample increases, the fraction of the agents that will remain active in the forthcoming round gets gradually smaller (\Cref{lemma:inclusion}), leading to \Cref{theorem:communication-single_item} and \Cref{corollary:optimal-single_item}. It should also be noted that the truncated communication does \emph{not} undermine the incentive compatibility of our mechanism, as implied by \Cref{proposition:IC} and \Cref{proposition:OSP}. Moving on to the multi-item auction with additive valuations, we employ a basic principle from information theory: encode the more likely events with fewer bits. This simple observation along with a property of our single-item auction allow us to substantially reduce the communication complexity when the auctions are executed in \emph{parallel}. 

Finally, we alter our approach for the design of a multi-unit auction with unit demand bidders. In contrast to a standard ascending auction, our idea is to broadcast in every round \emph{two} separate prices, a ``high'' price and a ``low'' price. Subsequently, the mechanism may simply recurse on the agents that reside between the two prices. The crux of this implementation is to design an algorithm that takes as input a \emph{small} number of bits, and returns prices that are tight upper and lower bounds on the final VCG payment. To this end, we design a novel algorithm that essentially performs \emph{stochastic binary search} on the tree that represents the valuation space; more precisely, in every junction, or decision, the branching is made based on a small \emph{sample} of agents, eliminating at every step half the points on the valuation space. From an algorithmic standpoint, this gives a simple procedure performing approximate \emph{selection} from an unordered list with very limited communication.

\subsection{Related Work}

Preference elicitation has received considerable attention in computational social choice \cite{10.5555/1641503.1641507,10.1145/1064009.1064018,10.5555/2484920.2485161,10.5555/2283396.2283445,10.1007/978-3-642-24873-3_11}. Segal \cite{SEGAL2007341} provided bounds on the communication required to realize a social choice rule through the notion of \emph{budget sets}, with applications in resource allocation tasks and stable matching. The boundaries of computational tractability and the strategic issues that arise in optimal preference elicitation were investigated by Conitzer and Sandholm for a series of voting schemes \cite{conitzer2002vote}, while the same authors established the worst-case number of bits required to execute common voting rules \cite{10.1145/1064009.1064018}. The trade-off between accuracy and information leakage in facility location games has been studied by Feldman et. al \cite{feldman2015voting}, investigating the behavior of truthful mechanisms with truncated input space---e.g., \emph{ordinal} information models. 

More recently, the trade-off between efficiency---in terms of \emph{distortion} as introduced by Procaccia and Rosenschein \cite{10.1007/11839354_23}---and communication was addressed by Mandal et al. \cite{DBLP:conf/nips/MandalP0W19} (see also \cite{DBLP:conf/aaai/AmanatidisBFV20}). Their results were subsequently improved in \cite{10.1145/3391403.3399510} using machinery from streaming algorithms, such as \emph{linear sketching} and $L_p$-\emph{sampling} \cite{DBLP:conf/soda/MonemizadehW10,DBLP:conf/focs/JayaramW18}. In similar spirit, \cite{10.5555/2540128.2540174,Bentert20:Comparing} address efficient preference elicitation in the form of \emph{top-$\ell$ elicitation}, meaning that the agents are asked to provide the length $\ell$ prefix of their ranking instead of their full ranking. This trade-off between efficiency and communication has also been an important consideration in the \emph{metric distortion framework}~\cite{Anshelevich15:Approximating}---a setting closely related to our work; e.g., see~\cite{Kizilkaya22:Plurality,Kempe20:Communication,Borodin22:Distortion,Anagnostides22:Metric}, as well as \cite{Anshelevich21:Distortion} for a comprehensive overview of that line of work.

Another important consideration in the literature relates to the interplay between communication constraints and incentive compatibility. In particular, Van Zandt \cite{10.2307/40005057} articulated conditions under which communication and incentive compatibility can be examined separately, while Reichelstein \cite{REICHELSTEIN198432}, and Segal and Fadel \cite{FADEL20091895} examined the communication overhead induced in truthful protocols, i.e. the communication cost of truthfulness. An additional prominent line of research studies mechanism design under communication constraints (see \cite{10.1086/676931} and references therein). Specifically, in a closely related to ours work, Blumrosen, Nisan and Segal \cite{DBLP:journals/corr/abs-1110-2733} considered the design of a single-item auction in a communication model in which every bidder could only transmit a limited number of bits. One of their key results was a $0.648$ social welfare approximation for $1$-bit auctions (every bidder could only transmit a \emph{single} bit to the mechanism) and uniformly distributed bidders. Further, the design of optimal---w.r.t. the obtained revenue---bid levels in an English auction was considered in \cite{soton260441}, assuming known prior distributions.

\section{Preliminaries}

\paragraph{Facility Location Games.} 

Consider a metric space $(\mathcal{M}, \dis(\cdot, \cdot))$, where $\dis : \mathcal{M} \times \mathcal{M} \to \mathbb{R}$ is a \emph{metric} (or distance function) on $\mathcal{M}$; i.e., for any $\mathbf{x}, \mathbf{y}, \mathbf{z} \in \mathcal{M}$, $\dis$ satisfies the following: (i) identity of indiscernibles: $\dis(\mathbf{x}, \mathbf{y}) = 0 \iff \mathbf{x} = \mathbf{y}$; (ii) symmetry: $\dis(\mathbf{x}, \mathbf{y}) = \dis(\mathbf{y}, \mathbf{x})$; and (iii) triangle inequality: $\dis(\mathbf{x}, \mathbf{y}) \leq \dis(\mathbf{x}, \mathbf{z}) + \dis(\mathbf{z}, \mathbf{y})$. Given as input an $n$-tuple $\mathcal{I} = (\mathbf{x}_1, \dots, \mathbf{x}_n)$, with $\mathbf{x}_i \in \mathcal{M}$, the $\ell$-facility location problem consists of allocating $\ell$ facilities on $\mathcal{M}$ in order to minimize the \emph{social cost}; more precisely, if $L$ is the finite the set of allocated facilities, the induced social cost w.r.t. the distance function $\dis$ is defined as
\begin{equation*}
    \SC(L, \dis) \triangleq \sum_{i=1}^n \min_{\mathbf{x} \in L} \dis(\mathbf{x}, \mathbf{x}_i);
\end{equation*}
that is, every point is assigned to its closest (allocated) facility. For notational simplicity, we omit the distance function and we simply write $\SC(L)$. With a slight abuse of notation, when $|L| = 1$, we will use $\SC(\mathbf{x})$ to represent the social cost of allocating a single facility on $\mathbf{x} \in \mathcal{M}$. In a mechanism design setting every point in the instance $\mathcal{I}$ is associated with a \emph{strategic} agent, and $\mathbf{x}_i$ represents her preferred \emph{private} location (e.g. her address); naturally, every agent $i$ endeavors to minimize her (atomic) distance from the allocated facilities. A mechanism is called \emph{strategy-proof} if for every agent $i$, and for any possible valuation profile, $i$ minimizes her distance by reporting her actual valuation; if the mechanism is \emph{randomized}, one is typically interested in strategy-proofness in expectation.

\paragraph{The Median Mechanism.} Posit a metric space $(\mathbb{R}^d, \| \cdot \|_1)$. The \textsc{Median} is a mechanism for the $1$-facility location problem which allocates a single facility on the median of the reported instance. For high-dimensional spaces, the median is defined coordinate-wise; more precisely, if $\mathcal{I} = (\mathbf{x}_1, \dots, \mathbf{x}_n)$ represents an arbitrary instance on $\mathbb{R}^d$, and we denote with $x_i^j$ the $j^{\text{th}}$ coordinate of $\mathbf{x}_i$ in some underlying coordinate system, 
\begin{equation*}
    \med(\mathcal{I}) \triangleq (\med(x_1^1, \dots, x_n^1), \dots, \med(x_1^d, \dots, x_n^d)).
\end{equation*}
In this context, the following properties have been well-establish in social choice, and we state them without proof.

\begin{proposition}
    \label{proposition:median-truthful}
The \textsc{Median} mechanism is strategy-proof.
\end{proposition}

\begin{proposition}
    \label{proposition:optimal-median}
The \textsc{Median} mechanism is optimal w.r.t. (i.e. minimizes) the social cost in the metric space $(\mathbb{R}^d, \| \cdot \|_1)$.
\end{proposition}

The median can also be defined beyond Euclidean spaces~\cite{SCHUMMER2002405}, as it will be discussed in more detail in \Cref{sec:treemedian}. It should also be noted that the median can be employed heuristically for the $\ell$-facility location problem when additional separability conditions are met; for example, the instance could correspond to residents of isolated cities, and it would be natural to assign one facility to a single ``cluster''.

\paragraph{Auctions.} In the domains we are studying in \Cref{section:auctions}, the valuation of an agent $i$ for any bundle of items $S$ is fully specified with the values $v_{i, j}$, for every item $j$; if a single item---and potentially multiple units of the same item---are to be disposed, we use $v_i$ for simplicity. Moreover, in \Cref{section:auctions} we employ a \emph{communication complexity} framework in order to analyze the measure of information elicited from the agents; thus, we need to assume that every value $v_{i, j}$ can be expressed with $k$ bits. We will mainly assume that $k$ is a parameter independent of the number of agents $n$, and one should imagine that a small constant $k$ (e.g. $32$ or $64$ bits) would suffice. In this context, we define the communication complexity of a mechanism to be the expected number of bits elicited from the participants during the interaction process, and we study the asymptotic growth of this function as $n \to \infty$. We will assume that an agent remains active in the auction only when \emph{positive} utility can be obtained. Thus, if the (monotonically increasing) \emph{announced price} for item $j$ coincides with the value $v_{i,j}$ of some agent $i$, we presume that $i$ will withdraw from the auction; we use this hypothesis to handle certain singular cases (e.g., all agents could have the same value for the item).

For the incentive compatibility analysis in \Cref{section:auctions} we will need to refine and extend the notion of strategy-proofness from facility location games. A mechanism will be referred to as \emph{strategy-proof} if truthful reporting is a \emph{universally} dominant strategy---a best response under any possible action profile and \emph{randomized} realization---for every agent.

\paragraph{Obvious Strategy-Proofness.} A strategy $s_i$ is \emph{obviously dominant} if for any deviating strategy $s_i'$, starting from any earliest information set where $s_i'$ and $s_i$ disagree, the best possible outcome from $s_i'$ is no better than the worst possible outcome from $s_i$. A mechanism is \emph{obviously strategy-proof} (OSP) if it has an equilibrium in obviously dominant strategies. Notice that OSP implies strategy-proofness, so it is a stronger notion of incentive compatibility \cite{561616}.

\paragraph{Ex-Post Incentive Compatibility.} We will also require a weaker notion of incentive compatibility; a strategy profile $(s_1, \dots, s_n)$ constitutes an \textit{ex-post Nash equilibrium} if the action $s_i(v_i)$ is a best response to every action profile $\mathbf{s}_{-i}(\mathbf{v}_{-i})$---for any agent $i$ and valuation $v_i$. A mechanism will be called \emph{ex-post incentive compatible} if sincere bidding constitutes an ex-post Nash equilibrium.

\paragraph{Additional Notation.} We will denote with $N = [n]$ the set of agents that participate in the mechanism. In a single parameter environment, the \emph{rank} of an agent corresponds to the index of her private valuation in ascending order (ties are broken arbitrarily according to some predetermined rule; e.g. lexicographic order) and indexed from $1$, unless explicitly stated otherwise. We will use the standard notation of $f(n) \sim g(n)$ if $\lim_{n \to + \infty} f(n)/g(n) = 1$ and $f(n) \lesssim g(n)$ if $\lim_{n \to + \infty} f(n)/g(n) \leq 1$, where $n$ will always be implied as the asymptotic parameter. For notational brevity, we will let $\binom{n}{m} = 0$ when $m > n$. Finally, $\| \mathbf{x} \|_1$ denotes the $L_1$ norm of $\mathbf{x} \in \mathbb{R}^d$, while $d$ will mainly represent the dimension of the underlying space.

\section{Sampling Approximation of the Median Mechanism}
\label{section:median}

Perhaps the most natural approximation of the $\textsc{Median}$ mechanism consists of taking a random sample of size $c$, and allocating a (single) facility to the median of the sample, as implemented in $\algomedian$ (Mechanism \ref{algo:approx_median}). Perhaps surprisingly, we will show that this simple approximation yields a social cost arbitrarily close to the optimal for the metric space $(\mathbb{R}^d, \| \cdot \|_1)$---for a sufficiently large sample. Our analysis commences with the median on the line, where our main contribution lies in \Cref{theorem:1-median_expectation}. Our approach is quite robust and yields analogous guarantees for the median defined on curves (\Cref{theorem:curve-median_expectation}) and the \emph{generalized median} on the metric space $(\mathbb{R}^d, \| \cdot \|_1)$ (\Cref{theorem:generalized-median_expectation}). We conclude this section by illustrating the barriers of sampling approximations when allocating a single facility on a \emph{tree metric}, as well as allocating \emph{multiple} facilities on the line.

\begin{algorithm}
\DontPrintSemicolon
\SetAlgoLined
\textbf{Input}: Set of agents $N$, accuracy parameter $\epsilon > 0$, confidence parameter $\delta > 0$\;
\textbf{Output}: $\mathbf{x} \in \mathbb{R}^d$ such that $\SC(\mathbf{x}) \leq (1 + \epsilon) \SC(\mathbf{x}_m)$, where $\mathbf{x}_m = \med(N)$\;
Set $c = \Theta(1/(\epsilon \delta)^2)$ to be the size of the sample\;
Let $S$ be a random sample of $c$ agents from $N$\;
\textbf{return} $\med(S)$\;
\caption{$\algomedian(N, \epsilon, \delta)$}
\label{algo:approx_median}
\end{algorithm}

In this section we do not have to dwell on incentive considerations given that our sampling mechanism $\algomedian$ directly inherits its truthfulness from the $\textsc{Median}$ mechanism (recall \Cref{proposition:median-truthful})---assuming of course that the domain admits a median.

\begin{proposition}
    $\algomedian$ is strategy-proof.
\end{proposition}

\subsection{Median on the Line}

In the following, we will tacitly consider an underlying arbitrary instance $\mathcal{I} = (x_1, \dots, x_n)$, with $x_i \in \mathbb{R}$ the (private) valuation---the preferred location---of agent $i$. To simplify the exposition, we will assume---without any loss of generality---that the number of agents $n$ is odd, with $n = 2\kappa + 1$ for some $\kappa \in \mathbb{N}$.

\paragraph{Sensitivity of the Median.}The first potential obstacle in approximating the $\textsc{Median}$ mechanism relates to the \emph{sensitivity} of the median. In particular, notice that the function of the median has an \emph{unbounded sensitivity}, that is, a unilateral deviation in the input can lead to an arbitrarily large shift in the output; more concretely, consider an instance with $n = 2\kappa + 1$ agents, and let $\kappa$ agents reside at $-l$ and $\kappa + 1$ agents at $+l$ for some large $l > 0$. If an agent from the rightmost group was to switch from $+l$ to $-l$, then the median would also relocate from $+l$ to $-l$, leading to a potentially unbounded deviation. It should be noted that in the regime of statistical learning theory, one technique to circumvent this impediment and ensure differential privacy revolves around the notion of \emph{smooth sensitivity}; e.g., see \cite{avellamedina2019differentially,brunel2020propose}. Our approach is based on the observation that the \emph{social cost} of the $\textsc{Median}$ inherently presents a sensitivity property. Formally, we establish the following lemma:

\begin{lemma}[Sensitivity of the $\textsc{Median}$]
    \label{lemma:sensitivity}
Let $x_m = \med(\mathcal{I})$ and $x \in \mathbb{R}$ some position such that $\epsilon \cdot n$ agents reside between $x$ and $x_m$, with $0 \leq \epsilon < 1/2$. Then, 
\begin{equation*}
    \SC(x) = \SC(x_m) \left( 1 + \mathcal{O} \left( \frac{4 \epsilon}{1 - 2\epsilon} \right) \right).
\end{equation*}
\end{lemma}

\begin{proof}
For the sake of presentation, let us assume that $x \geq x_m$. Consider the set $L$ containing the $\lfloor n/2 \rfloor - \epsilon \cdot n$ leftmost agents (ties are broken arbitrarily), and the set $R$ with the $\lfloor n/2 \rfloor - \epsilon \cdot n$ rightmost agents, leading to a partition as illustrated in \Cref{fig:img_1}. Now observe that if we transfer the facility from $x_m = \med(\mathcal{I})$ to $x$ the cumulative social cost of $L$ and $R$ remains invariant, i.e., 
\begin{equation*}
    \sum_{i \in R \cup L} |x_i - x_m| = \sum_{i \in R \cup L} |x_i - x|.
\end{equation*}
In other words, the increase in social cost incurred by group $L$ is exactly the social cost decrease of group $R$. Thus, it follows that 

\begin{equation}
    \label{eq:bound_1}
    \SC(x) \leq \SC(x_m) + 2\epsilon n |x - x_m|,
\end{equation}
where notice that this inequality is tight when the agents in the interval $(x_m, x)$ are accumulated arbitrarily close to $x_m$. Moreover, we have that 

\begin{equation}
    \label{eq:bound_2}
    \SC(x_m) \geq \sum_{i \in R} |x_i - x_m| \geq \left( \lfloor n/2 \rfloor - \epsilon n \right) |x - x_m|,
\end{equation}
and the claim follows from bounds \eqref{eq:bound_1} and \eqref{eq:bound_2}.
\end{proof}

\begin{figure}[!ht]
    \centering
    \includegraphics[scale=0.7]{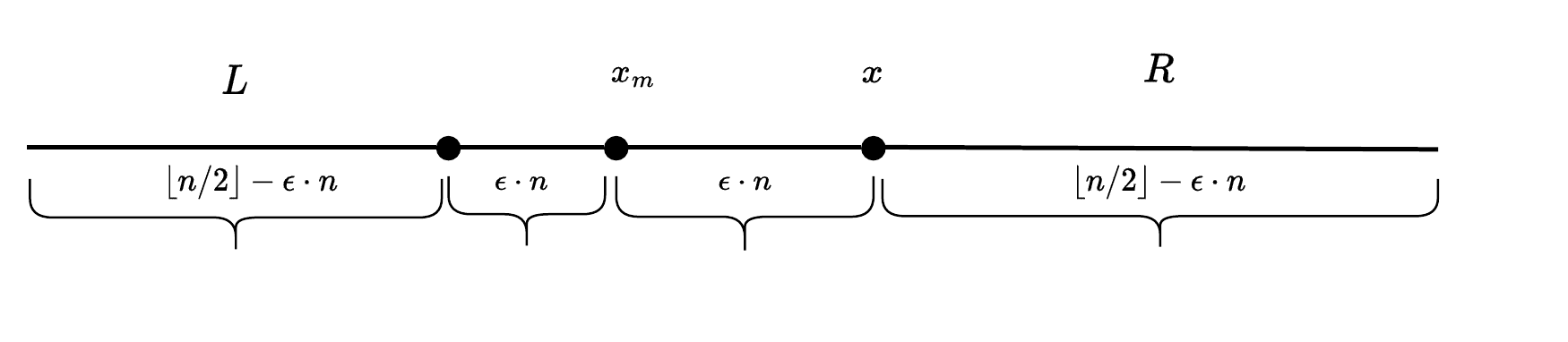}
    \caption{Partition of the agents on the line.}
    \label{fig:img_1}
\end{figure}

As a result, \Cref{lemma:sensitivity} implies that a unilateral deviation by a player can only lead to an increase of $\mathcal{O}(1/n)$ in the social cost. To put it differently, if an adversary corrupts arbitrarily a \emph{constant} number of reports, the increase in social cost will be asymptotically negligible. Indeed, even if the allocated facility lies arbitrarily far from the median, the induced social cost might still be near-optimal. 

Now let us assume---without loss of generality---that $c = 2\rho + 1$, for some $\rho \in \mathbb{N}$, where $c$ is the size of the sample; also recall that $n = 2\kappa + 1$. Motivated by \Cref{lemma:sensitivity}, our analysis will be \emph{oblivious} to the agents' locations on the line, but instead, our focus will be on characterizing the \emph{rank} of the sample's median---the relative order of the sample's median w.r.t. the entire instance; this approach will also allow us to directly obtain a guarantee in more general metric spaces. More precisely, consider a random variable $X_r$ that represents the rank---among the entire instance---of the sample's median, normalized in the domain $[-1, 1]$; for instance, if the sample's median happens to coincide with the median of the entire instance, then $X_r = 0$. The reason behind this normalization relates to our asymptotic characterization (\Cref{theorem:convergence}). Now fix a particular rank $i/\kappa$ in $[-1, 1]$. Notice that the number of configurations that correspond to the event $\{ X_r = i/\kappa\}$ is

$$ \binom{\kappa-i}{\rho} \binom{\kappa+i}{\rho}.$$
As a result, the probability mass function of $X_r$ can be expressed as follows:

\begin{equation}
    \label{eq:pmf}
    \pr\left[X_r = \frac{i}{\kappa}\right] = \ddfrac{\binom{\kappa - i}{\rho} \binom{\kappa + i}{\rho}}{\binom{2\kappa + 1}{2\rho + 1}}.
\end{equation}

It is interesting to notice the similarity of this expression to the probability mass function of a \emph{hypergeometric} distribution. We also remark that the normalization constraint of the probability mass function \eqref{eq:pmf} yields a well-known variant of the Vandermonde identity:  
\begin{equation*}
    \sum_{i = - \kappa}^{\kappa} \binom{\kappa - i}{\rho} \binom{\kappa + i}{\rho} = \sum_{i=0}^{2\kappa} \binom{i}{\rho} \binom{2\kappa -i }{\rho} = \binom{2\kappa + 1}{2 \rho + 1}.
\end{equation*}
For this reason, we shall refer to the distribution of $X_r$ as the $(\kappa, \rho)$-\emph{Vandermonde distribution}. Importantly, \Cref{lemma:sensitivity} tells us that the concentration of the Vandermonde distribution---for sufficiently large values of parameter $\rho$---suffices to obtain a near-optimal approximation with respect to the social cost. However, quantifying exactly the concentration of the Vandermonde distribution appears to be challenging.\footnote{Ariel Procaccia pointed out to us that there is a more elementary way to ``upper-bound'' the concentration of $X_r$; see \cite{DBLP:conf/innovations/BranzeiP15}, Lemma $1$. Yet, we remark that the aforementioned argument would only provide a guarantee with high probability, and not in expectation.} In light of this, our main insight---and the main technical contribution of this section---is an asymptotic characterization of this distribution. 

\begin{theorem}[Convergence of the Vandermonde Distribution]
    \label{theorem:convergence}
If we let $\kappa \to + \infty$ the $(\kappa,\rho)$-Vandermonde distribution converges to a continuous distribution with probability density function $\phi: [-1, 1] \to \mathbb{R}$, such that
\begin{equation}
    \label{eq:phi}
    \phi(t) = C(\rho) (1 - t^2)^{\rho},
\end{equation}
where
\begin{equation*}
    C(\rho) = B\left( \frac{1}{2}, \rho + 1 \right)^{-1} = \frac{(2\rho+1)!}{(\rho!)^2 2^{2\rho+1}}.
\end{equation*}
\end{theorem}

In the statement of the theorem, $B$ represents the \emph{beta function}. Recall that for $x, y \in \mathbb{R}_{> 0}$, the beta function is defined as 

\begin{equation}
    \label{eq:beta}
    B(x,y) = \int_{0}^1 t^{x-1}(1-t)^{y-1} dt = \frac{\Gamma(x) \Gamma(y)}{\Gamma(x + y)},
\end{equation}
where $\Gamma$ represents the \emph{gamma function}. One can verify the normalization constraint in \Cref{theorem:convergence} using a quadratic transform $u = t^2$ as follows:
\begin{equation*}
    \int_{-1}^1 (1 - t^2)^{\rho} dt = 2 \int_{0}^1 (1 - t^2)^{\rho} dt = \int_{0}^1 u^{ - \frac{1}{2}}(1 - u)^{\rho} du = B \left( \frac{1}{2}, \rho + 1\right).
\end{equation*}
Moreover, the final term can be expressed succinctly through the following lemma:

\begin{lemma}
If $\Gamma$ represents the gamma function and $n \in \mathbb{N}$, 
\begin{equation*}
    \Gamma\left(\frac{1}{2} + n \right) = \frac{(2n)!}{4^n n!} \sqrt{\pi}.
\end{equation*}
\end{lemma}

Before we proceed with the proof of \Cref{theorem:convergence}, we state an elementary result from real analysis. 

\begin{lemma}
    \label{lemma:integral}
    Let $f: [-1, 1] \to \mathbb{R}$ be an integrable function\footnote{The integrability here is implied in the standard Riemannian-Darboux sense.} and $x$ some number in $(-1, 1)$; then, 
    \begin{equation*}
        \lim_{n \to +\infty} \frac{x+1}{n} \sum_{i=1}^n f\left( -1 + i \cdot \frac{x+1}{n} \right) = \int_{-1}^x f(t) dt.
    \end{equation*}
\end{lemma}

\begin{proof}[Proof of \Cref{theorem:convergence}]
Take some arbitrary $x \in (-1, 1)$, let $\nu = \lfloor x \kappa \rfloor + \kappa + 1$, and consider a random variable $X_r$ drawn from a $(\kappa, \rho)$-Vandermonde distribution. It follows that

\begin{equation}
    \label{eq:con_1}
    \lim_{\kappa \to + \infty} \pr[X_r \leq x] = \lim_{\kappa \to + \infty} \sum_{i = -\kappa}^{\lfloor x \kappa \rfloor} \ddfrac{\binom{\kappa - i}{\rho} \binom{\kappa + i}{\rho}}{\binom{2\kappa + 1}{2\rho + 1}} = \lim_{n \to + \infty} \sum_{i = 1}^\nu \ddfrac{\binom{n - i}{\rho} \binom{i-1}{\rho}  }{\binom{n}{2\rho + 1}},
\end{equation}
where recall that $n = 2\kappa + 1$. Given that $n!/(n-j)! = \Theta_n(n^j), \forall j \in \mathbb{N}$, we can recast \eqref{eq:con_1} as 

\begin{align}
    \lim_{\kappa \to + \infty} \pr[X_r \leq x] &= \frac{(2\rho+1)!}{(\rho!)^2} \lim_{n \to + \infty} \frac{1}{n^{2\rho+1}} \sum_{i=1}^{\nu} \frac{(n-i)!}{(n-i-\rho)!} \frac{(i-1)!}{(i-1-\rho)!} \notag \\
    &= \frac{(2\rho+1)!}{(\rho!)^2} \lim_{n \to + \infty} \frac{1}{n^{2\rho+1}} \sum_{i=1}^{\nu} (n-i)^{\rho} i^{\rho} \label{eq:con_3}, 
\end{align}
where the last derivation follows by ignoring lower order terms. Finally, from \eqref{eq:con_3} we obtain that
\begin{align*}
    \lim_{\kappa \to + \infty} \pr[X_r \leq x] &= \frac{(2\rho+1)!}{(\rho!)^2} \lim_{n \to + \infty} \frac{1}{n} \sum_{i=1}^{\nu}  \left( \frac{i}{n} - \left( \frac{i}{n} \right)^2 \right)^{\rho} \\
    &= \frac{(2\rho+1)!}{(\rho!)^2} \lim_{\nu \to + \infty} \frac{x+1}{2\nu} \sum_{i = 1}^\nu \left(i \cdot \frac{x+1}{2\nu} -\left(i \cdot \frac{x+1}{2\nu} \right)^2 \right)^{\rho} \\
    &= \frac{(2\rho+1)!}{(\rho!)^2 2^{2\rho+1}} \lim_{\nu \to + \infty} \frac{x+1}{\nu} \sum_{i = 1}^\nu \left(2i \cdot \frac{x+1}{\nu} -\left(i \cdot \frac{x+1}{\nu} \right)^2 \right)^{\rho} \\
    &= \frac{(2\rho+1)!}{(\rho!)^2 2^{2\rho + 1}} \lim_{\nu \to + \infty} \frac{x+1}{\nu} \sum_{i = 1}^\nu \left( 1 - \left( -1 + i \cdot \frac{x+1}{\nu} \right)^2 \right)^{\rho} \\
    &= \frac{(2\rho + 1)!}{(\rho!)^2 2^{2\rho + 1}} \int_{-1}^x (1 - t^2)^{\rho} dt,
\end{align*}
where in the last line we applied \Cref{lemma:integral}, concluding the proof.
\end{proof}

Having established this asymptotic characterization, we are now ready to argue about the concentration of the induced distribution with respect to parameter $\rho$.

\begin{theorem}[Concentration]
    \label{theorem:concentration}
Consider a random variable $X$\footnote{It is interesting to note that $X$ is a \emph{sub-Gaussian} random variable \cite{Ledoux:1617672} with \emph{variance proxy} $\sigma^2 = \Theta(1/\rho)$; indeed, notice that $(1-t^2)^{\rho} \leq e^{- \rho t^2}$, with the bound being tight for $|t| \downarrow 0$. This observation leads to an alternative---and rather elegant---way to analyze the concentration of $X$.} with probability density function $\phi(t) = C(\rho) (1-t^2)^{\rho}$. Then, for any $\epsilon > 0$ and $\delta > 0$, there exists some $\rho_0 = \Theta(1/(\epsilon \delta)^2)$ such that $\forall \rho \geq \rho_0$,
\begin{equation*}
    \pr[|X| \geq \epsilon] \leq \delta.
\end{equation*}
\end{theorem}

\begin{proof}
Consider some $j \in \mathbb{N}$. The $j^{\text{th}}$ moment of $|X|$ can be expressed as 

\begin{equation*}
    \E[|X|^j] = C(\rho) \int_{-1}^{1} |t|^j (1-t^2)^{\rho} dt = 2 C(\rho) \int_{0}^1 t^j (1-t^2)^{\rho} dt.
\end{equation*}
Applying the quadratic transformation $u = t^2$ yields
\begin{equation}
    \label{eq:moment}
    \E[|X|^j] = \ddfrac{B \left( \frac{j}{2} + \frac{1}{2}, \rho + 1 \right)}{ B \left( \frac{1}{2}, \rho    + 1 \right)}.
\end{equation}

Recall from Stirling's approximation formula that $n! = \Theta\left(\sqrt{2\pi n} \left( \frac{n}{e} \right)^n \right)$; thus, we obtain that $C(\rho)$ grows as 
\begin{align*}
    C(\rho) = \frac{(2\rho + 1)!}{(\rho !)^2 2^{2\rho + 1}} = \Theta(\sqrt{\rho}).
\end{align*}

In particular, this along with \eqref{eq:moment} imply that $\E[|X|] = \Theta(1/\sqrt{\rho})$. Thus, if we apply Markov's inequality we obtain that 
\begin{equation*}
    \pr[|X| \geq \epsilon] = \mathcal{O} \left( \frac{1}{\epsilon \sqrt{\rho}} \right).
\end{equation*}
As a result, it suffices to select $\rho = \Theta(1/(\epsilon \delta)^2)$ so that $\pr[|X| \geq \epsilon] \leq \delta$; notice that tighter bounds w.r.t. the confidence parameter $\delta$ can be obtained if we apply Markov's inequality for higher moments of $|X|$ through \eqref{eq:moment}.
\end{proof}

\begin{corollary}
    \label{corollary:expectation}
Consider a random variable $X$ with probability density function $\phi(t) = C(\rho) (1-t^2)^{\rho}$. Then, $\E[|X|] = \Theta(1/\sqrt{\rho})$.
\end{corollary}

We are now ready to analyze the approximation ratio of our sampling median, both in terms of expectation and with high probability.

\begin{theorem}[Sampling Median on the Line]
    \label{theorem:1-median_expectation}
Consider a set of agents $N = [n]$ that lie on the metric space $(\mathbb{R}, | \cdot |)$. Then, for any $\epsilon > 0$, $\algomedian(N, \epsilon, \delta=1)$ takes a sample of size $c = \Theta(1/\epsilon^2)$ and yields in expectation a $1 + \epsilon$ approximation w.r.t. the optimal social cost of the full information $\textsc{Median}$, while $n \to + \infty$.
\end{theorem}

\begin{proof}
Consider a random variable $X$ with probability density function $\phi(t) = C(\rho) (1-t^2)^{\rho}$. \Cref{theorem:convergence} implies that $X$ corresponds to the rank of the sample's median with sample size $c = 2\rho + 1$, while $n \to + \infty$. Let $g : (0,1) \ni x \mapsto 2x/(1 - x)$; we know from \Cref{lemma:sensitivity} that the expected approximation ratio on the social cost is $1 + \mathcal{O}(\E[g(|X|)])$. But, it follows that 
\begin{align*}
    \E[g(|X|)] = 4 C(\rho) \int_{0}^1 \frac{t}{1 - t} (1-t^2)^{\rho} dt &\leq 8 C(\rho) \int_{0}^1 t (1-t^2)^{\rho-1} dt \\
    &= 8 \frac{2\rho + 1}{2\rho} C(\rho - 1) \int_{0}^1 t (1-t^2)^{\rho - 1} dt.
\end{align*}
Now notice that \Cref{corollary:expectation} implies that

\begin{equation*}
    C(\rho - 1) \int_{0}^1 t (1 - t^2)^{\rho -1} dt = \Theta\left( \frac{1}{\sqrt{\rho}} \right).
\end{equation*}

As a result, we have shown that the expected approximation ratio is $1 + \mathcal{O}(1/\sqrt{\rho})$, and taking $\rho = \Theta(1/\epsilon^2)$ concludes the proof.
\end{proof}

\begin{corollary}
    \label{corollary:1-median_whp}
Consider a set of agents $N = [n]$ that lie on the metric space $(\mathbb{R}, | \cdot |)$. Then, for any $\epsilon > 0$ and $\delta >0$, $\algomedian(N, \epsilon, \delta)$ takes a sample of size $c = \Theta(1/(\epsilon \delta)^2)$ and yields with probability at least $1 - \delta$ a $1 + \epsilon$ approximation w.r.t. the optimal social cost of the full information $\textsc{Median}$, while $n \to + \infty$.
\end{corollary}

\begin{proof}
The claim follows directly from \Cref{lemma:sensitivity}, \Cref{theorem:convergence,theorem:concentration}.
\end{proof}

\paragraph{Remark.} To supplement our asymptotic characterization, we refer the interested reader to \Cref{appendix:numerical} where we provide numerical bounds on the Vandermonde distribution for different values of $\kappa$ and $\rho$. Our experiments illustrate the rapid convergence of the Vandermonde distribution in terms of \emph{total variation distance}.

\subsubsection{Extension to the Median on Curves}

Here we give a slight extension of the previous guarantee when the agents lie on a \emph{curve}. More precisely, let $\mathcal{C}$ be a curve parameterized by a continuous function $\psi: [0,1] \to \mathbb{R}^d$. We will assume that $\mathcal{C}$ is \emph{simple} and \emph{open}, i.e. $\psi$ is injective in $[0,1]$; see an example in \Cref{fig:open_curve}. We also consider the distance between two points $A, B \in \mathcal{C}$ to be the \emph{length} of the induced (simple) sub-curve from $A$ to $B$, denoted with $\ell(A, B)$. 

Notice that any simple and open curve naturally induces a ranking---a \emph{total order}---over its domain; indeed, for any points $A, B \in \text{Im}(\psi) \equiv \mathcal{C}$, we let $A \preceq B \iff \psi^{-1}(A) \leq \psi^{-1}(B)$. Thus, we may define the median on the curve, which is strategy-proof and optimal w.r.t. the induced social cost. Importantly, our previous analysis is robust, and our methodology for the median on the line is directly applicable.

\begin{theorem}[Sampling Median on Curves]
    \label{theorem:curve-median_expectation}
Consider a set of agents $N = [n]$ that lie on the metric space $(\mathcal{C}, \ell(\cdot, \cdot))$, where $\mathcal{C}$ represents a simple and open curve on some subset of a Euclidean space. Then, for any $\epsilon > 0$, $\algomedian(N, \epsilon, \delta=1)$ takes a sample of size $c = \Theta(1/\epsilon^2)$ and yields in expectation a $1 + \epsilon$ approximation w.r.t. the optimal social cost of the full information $\textsc{Median}$, while $n \to + \infty$.
\end{theorem}

\begin{figure}[!ht]
    \centering
    \includegraphics[scale=0.45]{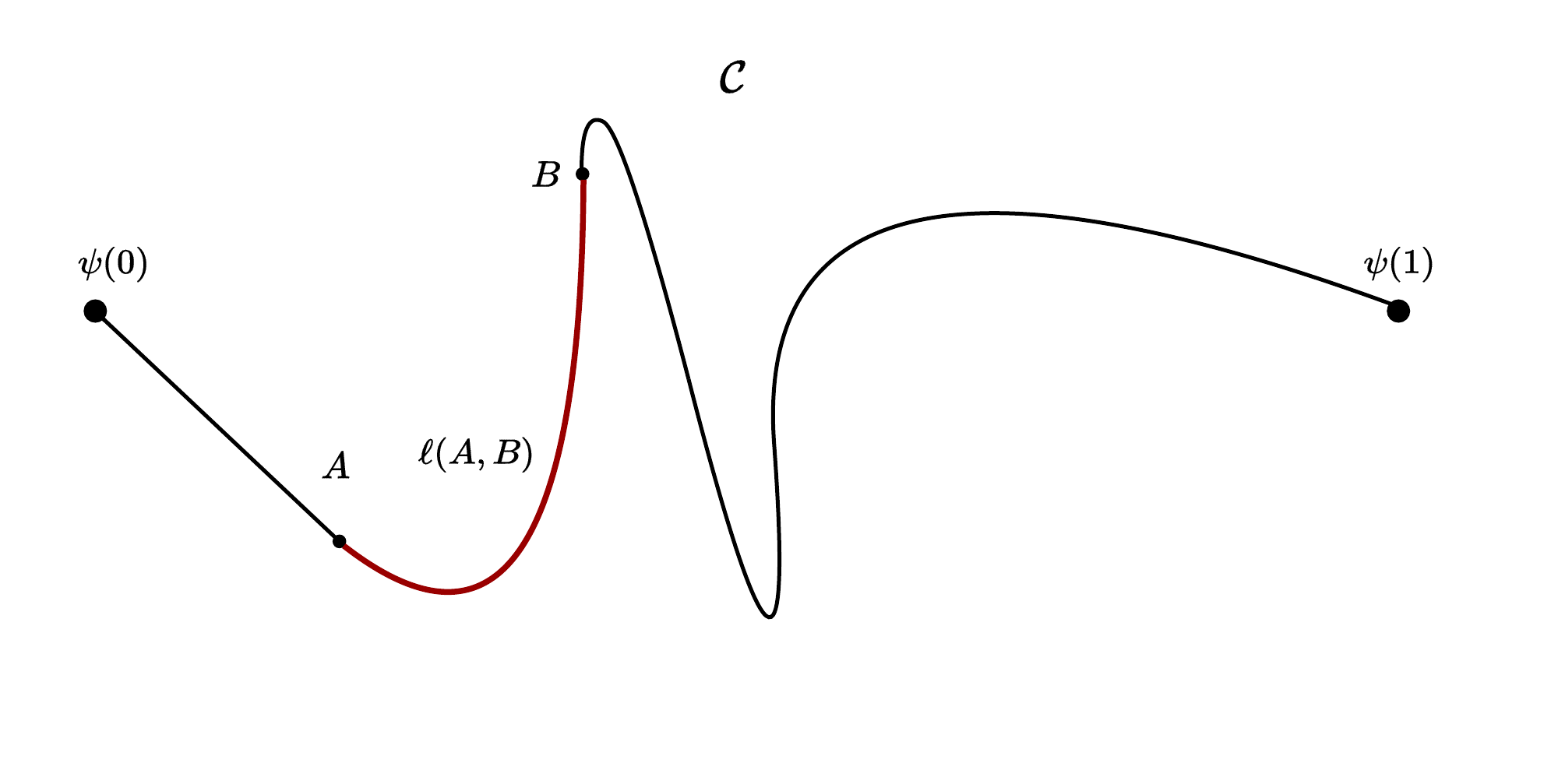}
    \caption{An example of a simple and open curve $\mathcal{C}$. The distance $\ell(A, B)$ between two points $A, B \in \mathcal{C}$ is defined as the length from $A$ to $B$.}
    \label{fig:open_curve}
\end{figure}

\subsection{High-Dimensional Median}

Next, in this subsection we extend our analysis for the high-dimensional metric space $(\mathbb{R}^d, \|\cdot \|_1)$. Specifically, consider some arbitrary instance $\mathcal{I} = (\mathbf{x}_1, \dots, \mathbf{x}_n)$, where $\mathbf{x}_i \in \mathbb{R}^d$; we will denote with $x_i^j$ the $j^{\text{th}}$ coordinate of $\mathbf{x}_i$ in some underlying coordinate system. Also, recall that the median of $\mathcal{I}$ is derived through the median in every coordinate, i.e.,

\begin{equation*}
    \med(\mathcal{I}) = (\med(x_1^1, \dots, x_n^1), \dots, \med(x_1^d, \dots, x_n^d)).
\end{equation*}

Naturally, the $\textsc{Median}$ takes as input some instance $\mathcal{I}$ and outputs the $\med(\mathcal{I})$. We establish the following theorem:

\begin{theorem}[High-Dimensional Sampling Median]
    \label{theorem:generalized-median_expectation}
Consider a set of agents $N = [n]$ that lie on the metric space $(\mathbb{R}^d, \| \cdot \|_1)$. Then, for any $\epsilon > 0$, $\algomedian(N, \epsilon, \delta=1)$ takes a sample of size $c = \Theta(1/\epsilon^2)$ and yields in expectation a $1 + \epsilon$ approximation w.r.t. the optimal social cost of the full information $\textsc{Median}$, while $n \to + \infty$.
\end{theorem}

\begin{proof}
Consider some facility at $\mathbf{x} \in \mathbb{R}^d$, with $\mathbf{x} = (x^1, \dots, x^d)$; the induced social cost of $\mathbf{x}$ can be expressed as 

\begin{equation*}
    \SC(\mathbf{x}) = \sum_{i=1}^n \| \mathbf{x}_i - \mathbf{x} \|_1 = \sum_{i=1}^n \sum_{j=1}^d |x_i^j - x^j| = \sum_{j=1}^d \sum_{i=1}^n |x_i^j - x^j|.
\end{equation*}

Let $\mathbf{X} = (X^1, \dots, X^d)$ be the (random) output of the $\algomedian$. For any $j \in [d]$, and with $c = \Theta(1/\epsilon^2)$, \Cref{theorem:1-median_expectation} implies that 
\begin{equation*}
    \E\left[\sum_{i=1}^n | x_i^j - X^j| \right] \leq (1 + \epsilon) \sum_{i=1}^n |x_i^j - x_m^j| + o_n(1),
\end{equation*}
where $\mathbf{x}_m = (x_m^1, \dots, x_m^d) = \med(\mathcal{I})$. Thus, by linearity of expectation we obtain that for $n \to \infty$, 

\begin{equation*}
    \E[\SC(\mathbf{X})] \leq (1 + \epsilon) \SC(\mathbf{x}_m) = (1 + \epsilon) \SC^*.
\end{equation*}
\end{proof}

Importantly, observe that even in the high-dimensional case it suffices to take $c = \Theta(1/\epsilon^2)$, \emph{independently} from the dimension of the space $d$, in order to obtain a guarantee in expectation. We also remark that it is straightforward to extend \Cref{corollary:1-median_whp} for the high-dimensional median, and recover a near-optimal allocation with high probability.

\subsection{Median on Trees}
\label{sec:treemedian}

In contrast to our previous positive results, our characterization breaks when allocating a facility on a general \emph{network}. In particular, consider an \emph{unweighted} \emph{tree} $G = (V, E)$, and assume that every node is occupied by a \emph{single} agent. One natural way to define the median on $G$ is by arbitrarily choosing a \emph{generalized median} for each path of the tree; however, as articulated by Vohra and Schummer \cite{SCHUMMER2002405}, for any two paths that intersect on an interval, it is crucial that the corresponding generalized medians must not contradict each other, a condition they refer to as \emph{consistency}. Providing a formal definition of consistency would go beyond the scope of our study; instead, we refer the interested reader to the aforementioned work.

Now imagine that the designer has \emph{no} prior information on the topology of the network, and will have to rely solely on the information extracted by the agents. Given that the graph might be vast, we consider a sample of nodes and we then construct the \emph{induced} graph by querying the agents in the sample\footnote{This model is analogous to the standard approach in \emph{property testing} \cite{goldreich_2017}.} (naturally, we posit that every agent knows her neighborhood). However, notice that the induced graph need not be a tree, and hence, it is unclear even how to determine the output of the sampling approximation. In fact, \emph{any} node in the subgraph may lead to a social cost far from the optimal.

\begin{figure}[!ht]
    \centering
    \includegraphics[scale=0.5]{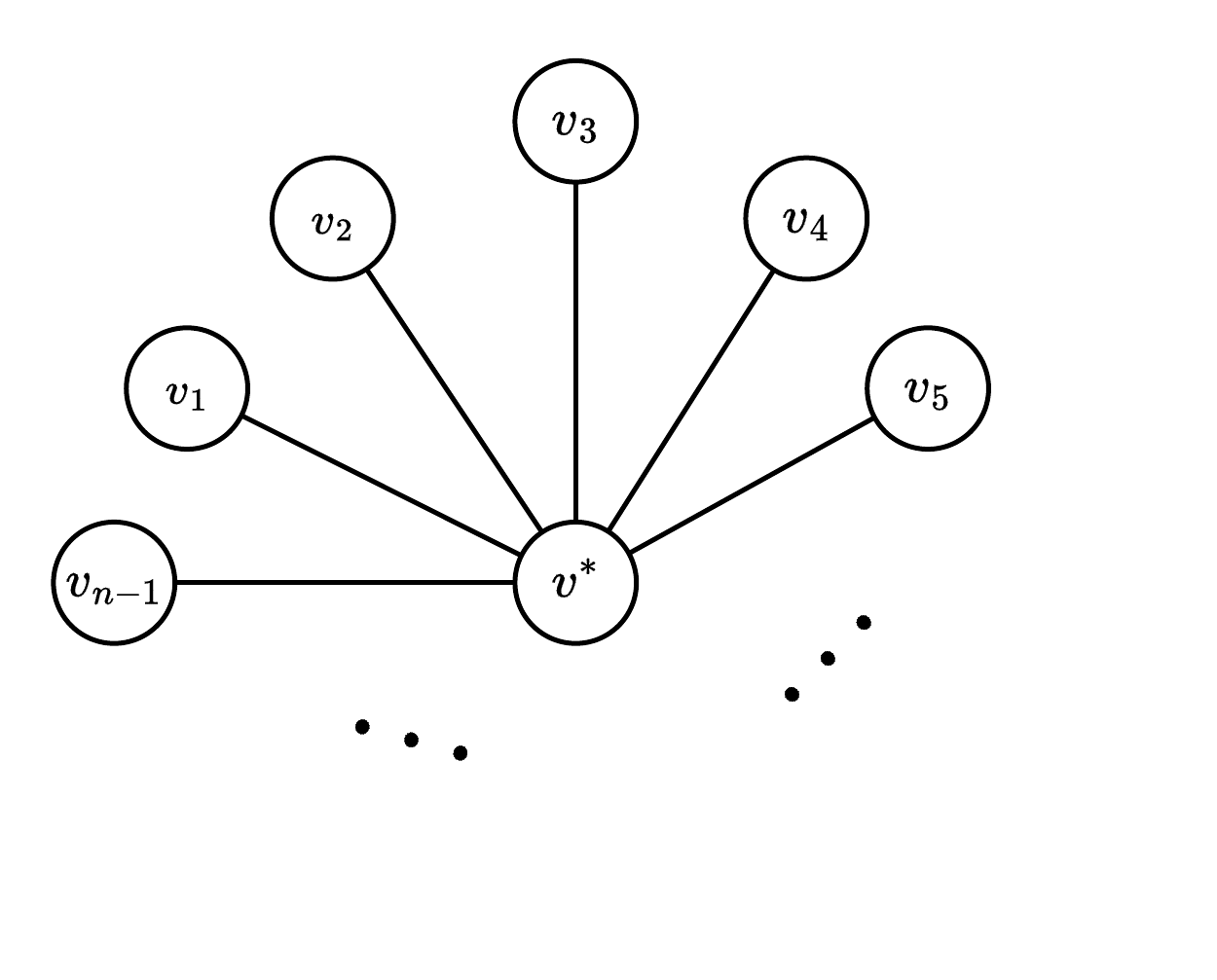}
    \caption{An unweighted ``star'' graph $G$ with $n$ nodes, with every node occupied by a \emph{single} agent. It is easy to verify that $v^*$ constitutes the unique median on $G$, satisfying the \emph{consistency} condition of Vohra and Schummer. Notice that allocating a facility on $v^*$ yields a social cost of $n-1$, while any other allocation leads to a social cost of $2n - 3$.}
    \label{fig:tree}
\end{figure}

\begin{proposition}
    \label{proposition:median-trees}
    Consider an unweighted star graph $G = (V, E)$ with $|V| = n$. Then, even if we take a sample of size $c = n/2$, every node in the sample will yield an approximation of $2 - 1/(n-1)$ w.r.t. the optimal social cost with probability $1/2$.
\end{proposition}

An approximation ratio of $2$ is trivial in the following sense: The $\textsc{RandomDictator}$ mechanism selects uniformly at random a single player and allocates the facility on her preferred position on the underlying metric space. An application of the triangle's inequality shows that $\textsc{RandomDictator}$ yields in expectation a $2$ approximation w.r.t. the optimal social cost.\footnote{On the other hand, it is easy to see that any deterministic dictatorship yields in the worst-case an $n-1$ approximation.} In that sense, augmenting the sample does not seem particularly helpful when the underlying metric space corresponds to a network.

\subsection{Sampling with Multiple Facilities}

Finally, we investigate the performance of a sampling approximation when allocating \emph{multiple} facilities. For the sake of simplicity, we posit a metric space $(\mathbb{R}, |\cdot|)$, and we consider the $\textsc{Percentile}$ mechanism, an allocation rule that assigns facilities on particular \emph{percentiles} of the input. More precisely, the $\textsc{Percentile}$ mechanism is parameterized by a sequence $r_1 < r_2 < \dots$, with every $r_j \in [n]$ corresponding to a rank of the input; if the instance $\mathcal{I} = (x_1, \dots, x_n)$ is given in increasing order, we allocate a facility $j$ on $x_{r_j}$ for every $j$. Naturally, the $\textsc{Median}$ can be classified in this family of mechanisms. Another prominent member is the $\textsc{TwoExtremes}$ mechanism, proposed by Procaccia and Tennenholtz \cite{10.1145/2542174.2542175} for the $2$-facility location problem. As the name suggests, this mechanism allocates two facilities at the minimum and the maximum reports of the instance, leading to an $n-2$ approximation ratio w.r.t. the optimal social cost (in fact, the $\textsc{TwoExtremes}$ is the only \emph{anonymous} and deterministic mechanism with bounded approximation ratio \cite{10.1145/2665005}). We remark that the $\textsc{Percentile}$ mechanism is always strategy-proof, while its approximation ratio w.r.t. the optimal social cost is generally unbounded.

Now consider the $\algopercentile$, simulating the $\textsc{Percentile}$ mechanism on a sample of size $c$; we will tacitly presume that at least $2$ facilities are to be allocated. Let us assume that the leftmost percentile $L$ contains at most $(1-\alpha) n$ agents, for some \emph{constant} $\alpha > 0$, and denote with $l$ the distance between $L$ and the complementary set of agents $R$; see \Cref{fig:percentile}. If we let the inner-distance between in $L$ and $R$ approach to $0$ and $l \to \infty$, we can establish the following:

\begin{proposition}
    \label{proposition:percentile}
    There are instances for which even with a sample of size $c = \alpha n = \Theta(n)$ the $\algopercentile$ has in expectation an unbounded approximation ratio w.r.t. the social cost of the full information $\textsc{Percentile}$ mechanism.
\end{proposition}

\begin{figure}[!ht]
    \centering
    \includegraphics[scale=0.5]{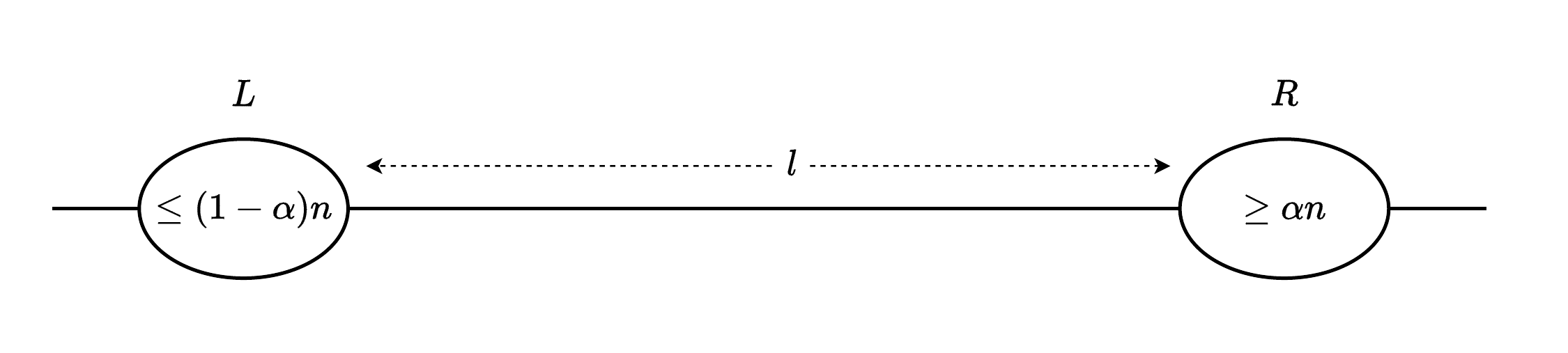}
    \caption{An instance that corresponds to \Cref{proposition:percentile}, as $l \to \infty$; given that at least two facilities are to be allocated, any mechanism with bounded social cost has to allocate facilities on both $L$ and $R$.}
    \label{fig:percentile}
\end{figure}

Indeed, for the instance we described the social cost of the full information mechanism approaches to $0$, as it always allocates a facility on $L$, and at least one facility on $R$. In contrast, there will be a positive probability---albeit exponentially small---that the approximation mechanism will fail to sample an agent from $L$. Moreover, the same limitation (\Cref{proposition:percentile}) applies for an additive approximation, instead of a relative one. Perhaps, it would be interesting to examine the performance of a sampling approximation if we impose additional restrictions on the instance, such as \emph{stability} conditions.

\paragraph{Remark.} One of our insights is that unlike the \emph{outcome} of the underlying mechanism, the social cost often presents a sensitivity property. This observation could be of independent interest in analyzing sampling approximations in voting. We explore this idea for the \emph{plurality rule} in \Cref{appendix:plurality}.

\section{Optimal Communication in Simple Auctions}
\label{section:auctions}

In this section, we study a series of environments in the regime of auction theory. For every instance, we develop a mechanism which asymptotically minimizes the communication complexity, i.e. the numbers of bits elicited from the participants. We commence with the design of a single-item auction, and we gradually extend our techniques to cover more general domains.

\subsection{Single-Item Auction}

This subsection presents our ascending auction for disposing a \emph{single} and \emph{indivisible} item. On a high level, instead of updating the price in a static manner through a \emph{fixed} increment, we propose a simple \emph{adaptive} mechanism. Before we proceed with the description and the analysis of our mechanism, we wish to convey some intuition for our implementation. Specifically, as a thought experiment, imagine that every valuation $v_i$ is drawn from some arbitrary distribution $\mathcal{D}$, and assume---rather unreasonably---that we have access to this distribution. In this context, one possible approach to minimize the communication would be to determine a threshold value $T_h$ such that $\pr_{v \sim \mathcal{D}}[v \geq T_h] \leq \epsilon$, for some small $\epsilon > 0$. The auctioneer could then simply broadcast at the first round of the auction the price $p := T_h$, and with high probability the agents above the threshold would constitute only a small fraction of the initial population. Moreover, the previous step could be applied recursively for the distribution $\mathcal{D}$ conditioned on $v \geq T_h$, until only a few agents remain active. Our proposed mechanism will essentially mirror this though experiment, but without any distributional assumptions, or indeed, any prior knowledge. 

Specifically, we consider some black-box \emph{deterministic} algorithm $\mathcal{A}$ that \emph{faithfully} simulates a second-price auction. Namely, $\mathcal{A}$ takes as an input a set of agents $S$ and returns a tuple $(w, p)$: $w$ is the agent with the highest valuation among $S$ (ties are broken arbitrarily), and $p$ corresponds to the second highest valuation. Of course, $\mathcal{A}$ only \emph{simulates} a second-price auction, without actually allocating items and imposing payments. 

In every round, our mechanism selects a random sample from the active agents,\footnote{Naturally, we assume that the sample contains at least $2$ agents, so that the second-price rule is properly defined.} and simulates through algorithm $\mathcal{A}$ a sub-auction. Then, the ``market-clearing price'' in the sub-auction is announced as the price of the round, and this process is then repeated iteratively. The pseudocode for our mechanism ($\algoascending$) is given in \ref{algo:ascending_auction}.

\begin{algorithm}
\DontPrintSemicolon
\SetAlgoLined
\textbf{Input}: Set of agents $N$, algorithm $\mathcal{A}$ which simulates a second-price auction, parameter $\epsilon > 0$\;
\textbf{Output}: VCG outcome (Winner \& Payment) \;
\While{$|N| > c$ }{
  Let $S$ be a random sample of $c = \Theta(1/\epsilon^2)$ agents from $N$\;
  Set $w :=$ winner in $\mathcal{A}(S)$\; 
  Announce $p:=$ payment in $\mathcal{A}(S)$ \;
  Update the active agents: $N := \{ i \in N \setminus S : v_i > p \} \cup \{ w\} $\;
 }
 \eIf{$|N| = 1$}{
 \textbf{return} $(w, p)$\;
 }{
 \textbf{return} $\mathcal{A}(N)$\;
 }
\caption{$\algoascending(N, \mathcal{A}, \epsilon)$}
\label{algo:ascending_auction}
\end{algorithm}

Interestingly, our mechanism induces a format that couples an ascending auction with the auction simulated by $\mathcal{A}$. We shall establish the following properties:

\begin{proposition}
    \label{proposition:VCG}
    Assuming truthful reporting, $\algoascending$ returns---with probability $1$---the VCG outcome.
\end{proposition}

\begin{proof}
First, notice that in any iteration of the while loop only agents that are below or equal to the second highest valuation will withdraw from the auction. Now consider the case where upon exiting the while loop only a \emph{single} agent $w$ remains in the set of active agents $N$. Then, it follows that the announced price $p$ in the final round---which by construction coincides with the valuation of some player---exceeds the valuation of every player besides $w$. Thus, by definition, the outcome implements the VCG rule. Moreover, if after the last round $2 \leq |N| \leq c$, the claim follows given that $\mathcal{A}$ faithfully simulates a second-price auction.
\end{proof}

\begin{proposition}
    \label{proposition:IC}
    If $\mathcal{A}$ simulates a sealed-bid auction, $\algoascending$ is ex-post incentive compatible.\footnote{In fact, $\algoascending$ (with $\mathcal{A}$ implemented as a sealed-bid auction) is dominant strategy incentive compatible if the sequence of announced prices is non-decreasing; this property can be guaranteed if the ``market-clearing price'' in the sub-auction serves as a \emph{reserved price}. Otherwise, truthful reporting is \emph{not} necessarily a dominant strategy. For example, assume that every agent $i$, besides some agent $j$, commits to the following---rather ludicrous---strategy: $i$ reports truthfully, unless $i$ remains active with at most $2c$ other agents; in that case, $i$ will act as if her valuation is $0$. Then, the best response for $j$ is to act as if her valuation is $\infty$.}
\end{proposition}

\begin{proof}
Consider any round of the auction and some agent $i$ that has been selected in the sample $S$; we identify two cases. First, if $v_i \geq v_j, \forall j \in S \setminus \{i\}$, sincere reporting clearly constitutes a best response for $i$. Indeed, notice that since $\mathcal{A}$ simulates a second-price auction, the winner in the sub-auction does not have any control over the announced price of the round. In the contrary case, agent $i$ does not have an incentive to misreport and remain active in the auction given that the final payment will always be greater or equal to her valuation---observe that the announced price always increases throughout the auction. Next, let $p$ the payment in $\mathcal{A}(S)$ and $i \notin S$. It follows that if $v_i \leq p$ then a best response for $i$ is to withdraw from the auction, while if $v_i > p$ then $i$'s best response is to remain active in the forthcoming round.
\end{proof}

\begin{proposition}
    \label{proposition:OSP}
    If $\mathcal{A}$ simulates an English auction, $\algoascending$ is obviously strategy-proof (OSP).
\end{proposition}

\begin{proof}
Notice that the induced mechanism performs a standard English auction, but instead of interacting with \emph{every} agent in a given round, we ``ascend'' in a small sample; only when a single agent remains active we broadcast the current price to the rest of the agents. Now consider some agent $i$ that participates in the sub-auction. If the current price is below her valuation $v_i$, then the best possible outcome from quitting is no better than the worst possible outcome from staying in the auction. Otherwise, if the price is above $v_i$, then the best possible outcome from staying in the auction is no better than the worst possible outcome from withdrawing in the current round. Indeed, notice that the announced price can only increase throughout the auction. Of course, the same line of reasoning applies for a round of interaction with the entire set of active agents; essentially, the claim follows from the OSP property of the English auction.
\end{proof}

Next, we analyze the communication complexity of our proposed auction. Naturally, we have to assume that the valuation space is discretized, with $k$ bits being sufficient to express any valuation. The first thing to note is a trivial lower bound on the communication. 

\begin{fact}[Communication Lower Bound]
    \label{fact:lower_bound}
Every mechanism that determines the agent with the highest valuation---with probability $1$---must elicit at least $n$ bits.
\end{fact}

Importantly, we will show that this lower bound can be asymptotically recovered. In particular, observe that as the size of sample increases, the fraction of agents that will choose to remain active---at least in the forthcoming round---gradually diminishes; the following lemma makes this property precise.

\begin{lemma}[Inclusion Rate]
    \label{lemma:inclusion}
    Let $X_a$ be a random variable representing the proportion of agents that will remain active in a given round of the $\algoascending$ with sample size $c$; then,
    \begin{equation*}
        \E[X_a] \lesssim \frac{2}{c + 1}.
    \end{equation*}
\end{lemma}

As a result, the size of the sample $c$ allows us to calibrate the number of agents that we wish to include in the following round --- the inclusion rate. Before we proceed with the proof of \Cref{lemma:inclusion}, we first state some standard asymptotic formulas.

\begin{fact}
\label{fact:bc-asympt}
\begin{equation*}
    \binom{n}{c} \sim \frac{n^c}{c!}.
\end{equation*}
\end{fact}

\begin{fact}
\label{fact:ps-asympt}
\begin{equation*}
    \sum_{i=1}^n i^p \sim \frac{n^{p+1}}{p+1} \sim \sum_{i=1}^{n} (i-1)^{p}.
\end{equation*}
\end{fact}

\begin{proof}[Proof of \Cref{lemma:inclusion}]
Consider some round of the $\algoascending$ with $n$ active agents (here we slightly abuse notation given that $n$ corresponds to the initial number of agents). Let us denote with $X_r$ the rank---in the domain $[n]$---of the player with the second highest valuation (recall that ties are broken arbitrarily according to some \emph{fixed} order) in the sample. We will show that 
\begin{equation*}
    \E[X_r] \sim n \frac{c-1}{c+1}.
\end{equation*}
Indeed, simple combinatorial arguments yield that the probability mass function of $X_r$ can be expressed as 
\begin{equation*}
     \Pr[X_r = i] = \ddfrac{ \binom{n-i}{1} \binom{i-1}{c-2}}{\binom{n}{c}}.
\end{equation*}
As a result, it follows that 
\begin{align*}
    \E[X_r] = \sum_{i=1}^n i \Pr[X_r = i] &\sim \frac{c!}{n^c} \sum_{i=1}^n i (n-i) \binom{i-1}{c-2} \\
    &= \frac{c!}{n^c} \left( n \sum_{i=1}^n i \binom{i-1}{c-2} - \sum_{i=1}^n i^2 \binom{i-1}{c-2} \right) \\
    &\sim \frac{c!}{n^c} \left( n \sum_{i=1}^n i \binom{i}{c-2} - \sum_{i=1}^n i^2 \binom{i}{c-2} \right) \\
    &\sim \frac{c!}{n^c} \left( n \sum_{i=1}^n \frac{i^{c-1}}{(c-2)!} - \sum_{i=1}^n \frac{i^{c}}{(c-2)!} \right) \\
    &\sim \frac{c (c-1)}{n^c} \left( \frac{n^{c+1}}{c} - \frac{n^{c+1}}{c+1} \right) \\
    &= n \frac{c-1}{c+1},
\end{align*}
where we applied the asymptotic bounds from \Cref{fact:bc-asympt} and \Cref{fact:ps-asympt}; also note that we ignored lower order terms in the third and fourth line. Finally, the proof follows given that $X_a \leq (n - X_r)/n$; the inequality here derives from the fact that multiple agents could have the same valuation with the agent with rank $X_r$, and we assumed that such agents will quit. 
\end{proof}

Next, we are ready to analyze the communication complexity of our mechanism. In the following theorem, we implicitly assume that the agents report sincerely.

\begin{theorem}
    \label{theorem:communication-single_item}
    Let $Q$ be the (worst-case) communication complexity of some deterministic algorithm $\mathcal{A}$ that faithfully simulates a second-price auction, and $N = [n]$ a set of agents. Moreover, for any $\epsilon > 0$ and $c = \Theta(1/\epsilon^2)$, denote by $t(n; c, k)$ the expected communication complexity of $\algoascending(N, \mathcal{A}, \epsilon)$. Then, 
    \begin{equation}
        \label{eq:comm_Q}
        t(n; c, k) \lesssim (1 + \epsilon) n + Q(c; k) \log n.
    \end{equation}
    
\end{theorem}

\begin{proof}
Consider some round of the $\algoascending$ with $n$ active agents; as in \Cref{lemma:inclusion}, let us denote with $X_a$ the proportion of agents that will remain active in the following round of the auction. If $T$ represents the (randomized) communication complexity of our mechanism, we obtain that 

\begin{equation}
\label{equation:recursion}
    \E[T(n; c, k)] = \E[T(n X_a; c, k)] + Q(c; k) + n - c.
\end{equation}

There are various ways to bound randomized recursions of such form; our analysis will leverage the concentration of $X_a$. In the sequel, we will tacitly assume that the agents' valuations are pairwise distinct, as this yields an upper bound on the actual communication complexity (in our case, ties can only truncate communication). We will first establish that 
\begin{equation}
    \label{eq:var-X_a}
    \V[X_a]  \sim \frac{2(c-1)}{(c+2)(c+1)^2}.
\end{equation}

Indeed, consider the random variable $X_r$ that represents the rank of the agent with the second highest valuation in the sample. Analogously to the proof of \Cref{lemma:inclusion}, it follows that 

\begin{align*}
    \E[X_r^2] = \sum_{i=1}^n i^2 \Pr[X_r = i] &\sim \frac{c!}{n^c} \sum_{i=1}^n i^2 (n-i) \binom{i-1}{c-2} \\
    &\sim \frac{c!}{n^c} \left( n \sum_{i=1}^n i^2 \binom{i}{c-2} - \sum_{i=1}^n i^3 \binom{i}{c-2} \right) \\
    &\sim \frac{c!}{n^c} \left( n \sum_{i=1}^n \frac{i^{c}}{(c-2)!} - \sum_{i=1}^n \frac{i^{c+1}}{(c-2)!} \right) \\
    &\sim \frac{c (c-1)}{n^c} \left( \frac{n^{c+2}}{c+1} - \frac{n^{c+2}}{c+2} \right) \\
    &= n^2 \frac{c(c-1)}{(c+2)(c+1)}.
\end{align*}

As a result, \eqref{eq:var-X_a} follows given that $\V[X_r] = \E[X_r^2] - (\E[X_r])^2$ and $\V[X_a] = \V[X_r]/n^2$; notice that under the assumption that the valuations are pairwise distinct, it follows that $X_a = (n - X_r)/n$. For notational simplicity, let us denote with $\mu = \E[X_a]$ and $\sigma = \sqrt{\V[X_a]}$. Chebyshev's inequality implies that $\pr[|X_a - \mu| \geq \sqrt{c} \sigma] \leq 1/c$. It is also easy to see that $\mu + \sqrt{c} \sigma \leq 4/\sqrt{c}$; hence, with probability at least $1 - 1/c$, $n X_a \leq 4n/\sqrt{c}$. Consequently, \eqref{equation:recursion} gives that 

\begin{equation}
    \label{eq:recursion_bound}
        t(n; c, k) \lesssim \left( 1 - \frac{1}{c} \right) t\left( \frac{4 n }{\sqrt{c}}; c, k \right) + \frac{1}{c} t(n; c, k) + n + Q(c; k),
    \end{equation}
where we used the fact that $t(n; c, k)$ is decreasing w.r.t. $n$. Moreover, \eqref{eq:recursion_bound} can be recast as 

\begin{equation}
    t(n; c, k) \lesssim t\left( \frac{4n}{\sqrt{c}}; c, k \right) + \frac{c}{c-1} n + \frac{c}{c-1} Q(c; k).
\end{equation}

Now consider any small $\epsilon > 0$, and let $c = \Theta(1/\epsilon^2)$;\footnote{We do not claim that our analysis w.r.t. the size of the sample is tight; the rather crude bound $c = \Theta(1/\epsilon^2)$ is an artifact of our analysis, but nonetheless it will suffice for our purposes. } from the previous recursion we obtain that 
\begin{equation*}
    t(n; c, k) \lesssim (1 + \epsilon) n + Q(c; k) \log n,
\end{equation*}
as desired.
\end{proof}

In particular, if $\mathcal{A}$ simulates a sealed-bid auction $Q(c; k) = c \cdot k$, while if $\mathcal{A}$ simulates an English auction $Q(c; k) = c \cdot 2^k$; indeed, implementing a \emph{faithful} second-price auction through a standard ascending format necessitates covering the entire valuation space, i.e. $2^k$ potential prices. \Cref{theorem:communication-single_item} implies that the size of the sample $c$ induces a trade-off between two terms: As we augment the size of the sample $c$ we truncate the first term in \eqref{eq:comm_Q}---most agents withdraw from a given round---at the cost of increasing the simulation of the sub-auction $\mathcal{A}$ ---the second term in \eqref{eq:comm_Q}. Returning to our earlier thought experiment where we had access to the distribution over the valuations, the term $Q(c; k) \log n$ is essentially the overhead which we incur given that we do not possess any prior information. Yet, if $k$ does not depend on $n$ and we examine the asymptotic growth of the expected communication complexity w.r.t. the number of agents $n$, we obtain the following:

\begin{corollary}[Single-Item Auction with Optimal Communication]
    \label{corollary:optimal-single_item}
    Consider some algorithm $\mathcal{A}$ that faithfully simulates a second-price auction, and $N = [n]$ a set of agents. Moreover, for any $\epsilon > 0$ and $c = \Theta(1/\epsilon^2)$, denote by $t(n; c, k)$ the expected communication complexity of the $\algoascending(N, \mathcal{A}, \epsilon)$. Then, if $k$ is a constant independent of $n$,
    \begin{equation*}
        t(n; c, k) \lesssim (1 + \epsilon) n.
    \end{equation*}
\end{corollary}

\paragraph{Remark.} It is important to point out that the elicitation pattern in our proposed mechanism ($\algoascending$) is highly \emph{asymmetrical}. Indeed, while most of the agents will be eliminated after the first round, having only revealed a \emph{single} bit from their valuations, the agents who are "close" to winning the item will have to disclosure a substantial amount of information; arguably, this property is desirable. Notice that this is in stark contrast to a standard English auction in which every withdrawing bidder approximately reveals her valuation. 

\subsection{Multi-Item Auction with Additive Valuations}

As an extension of the previous setting, consider than the auctioneer has to allocate $m$ (indivisible) items to $n$ agents, with the valuation space being \emph{additive}; that is, for every agent $i$ and for any bundle of items $S \neq \emptyset$,
\begin{equation*}
    v_i(S) = \sum_{j \in S} v_{i,j},
\end{equation*}
where recall that $v_{i, j}$ represents the value of item $j$ for agent $i$. Naturally, we are going to employ an $\algoascending$ for every item. It should be clear that -- by virtue of \Cref{proposition:VCG} -- the induced mechanism implements with probability $1$ the VCG outcome; every item is awarded to the agent who values it the most, and the second-highest valuation for that particular item is imposed as the payment. Moreover, \Cref{proposition:IC} implies the following:

\begin{proposition}
    Implementing for every item $\algoascending$ yields an ex-post incentive compatible multi-item auction with additive valuations.
\end{proposition}

Our main insight in this domain is that a \emph{simultaneous} implementation can lead to a much more communication-efficient interaction process.

\paragraph{Sequential Implementation.} First, assume that we have to perform an independent and separate auction for each item. Then, \Cref{fact:lower_bound} implies that our mechanism has to elicit at least $n \cdot m$ bits. As in the single-item setting, we can asymptotically match this lower bound. 

\begin{proposition}
    Consider a set of agents $N = [n]$ with additive valuations for $m$ (indivisible) items, and denote by $t(n; m, c, k)$ the expected communication complexity of implementing for every item $\algoascending$. Then, for any $\epsilon > 0$ and with $k$ assumed a constant independent of $n$,
    \begin{equation*}
        t(n; m, c, k) \lesssim (1 + \epsilon) n m .
    \end{equation*}
\end{proposition}

\paragraph{Simultaneous Implementation.} On the other hand, we will show that the communication complexity can be substantially reduced when the $m$ auctions are performed in \emph{parallel}. Specifically, our approach employs some ideas from information theory in order to design a more efficient encoding scheme. More concretely, let us first describe the general principle. Consider a discrete random variable that has to be \emph{encoded} and subsequently transmitted to some receiver; it is well understood in coding theory that the values of the random variable which are more likely to be observed have to be encoded with relatively fewer bits, so that the communication complexity is minimized in expectation. 

Now the important link is that in the $\algoascending$ with a large sample size $c$ a random agent will most likely withdraw from a given round. Thus, we consider the following encoding scheme: An agent $i$---remaining active in at least one of the $m$ auctions---will transmit a bit $0$ in the case of withdrawal from \emph{all} the auctions; otherwise, $i$ may simply transmit an $m_i$-bit vector that indicates the auctions that $i$ wishes to remain active, where $m_i \leq m$ is the number of auctions in which $i$ is still active. Although the latter part of the encoding is clearly sub-optimal---given that we have encoded events with substantially different probabilities with the same number of bits, it will be sufficient for our argument. Consider a round of the parallel implementation with $n$ agents and let $p$ be the probability that a random agent will withdraw from every auction in the current round. Given that every player is active in at most $m$ auctions it follows from the union bound that $1 - p \lesssim 2m/(c+1)$. Thus, if $B$ represents the total number of bits transmitted in the round, we obtain that 

\begin{equation*}
\E[B] = n \left( 1 \cdot p + m \cdot (1-p) \right) \lesssim n \left( \left( 1 - \frac{2m}{c+1} \right) + m \left( \frac{2m}{c+1} \right) \right).    
\end{equation*}

As a result, for every $\delta > 0$ and size of sample $c = \Theta(m^2/\delta)$, we have that $\E[B] \leq n (1 + \delta)$. Also note that in expectation only a small fraction of agents will ``survive'' in a given round of the parallel auction---asymptotically at most $2m/(c+1)$. Thus, similarly to \Cref{corollary:optimal-single_item} we can establish the following theorem: 

\begin{theorem}[Simultaneous Single-Item Auctions] 
    \label{theorem:simultaneous}
    Consider some set of agents $N = [n]$ with additive valuations for $m$ indivisible items, and assume that $k$ and $m$ are constants independent of $n$. There exists an encoding scheme such that if $t(n; m, c, k)$ is the expected communication complexity of implementing in parallel an $\algoascending$ for every item, then for any $\epsilon > 0$ and for sufficiently large $c = c(\epsilon, m)$,
    \begin{equation*}
        t(n; m, c, k) \lesssim (1 + \epsilon) n.
    \end{equation*}
\end{theorem}

\subsection{Multi-Unit Auction with Unit Demand}

Finally, we design a multi-unit auction where $m$ units of the same item are to be disposed to $n$ \emph{unit demand} bidders; naturally, we are interested in the non-trivial case where $m \leq n$. We shall consider two canonical cases. 

First, let us assume that the number of units $m$ is a small \emph{constant}. Then, we claim that our approach in the single-item auction can be directly applied. Indeed, we propose an ascending auction in which at every round we invoke some algorithm $\mathcal{A}$ that simulates the VCG outcome---i.e., $\mathcal{A}$ identifies the $m$ agents with the highest valuations, as well as the $(m+1)$-highest valuation as the payment---for a random sample. Next, the ``market-clearing price'' in the sub-auction is announced in order to ``prune'' the active agents. As a result, we can establish guarantees analogous to \Cref{proposition:VCG,proposition:IC} and \Cref{corollary:optimal-single_item}; the analysis is similar to the single-item auction, and is therefore omitted.

Our main contribution in this subsection is to address the case where $m = \gamma \cdot n$, for some \emph{constant} $\gamma \in (0,1)$. Specifically, unlike a standard English auction, our idea is to broadcast in every round \emph{two} separate prices; the agents who are above the high price $p_h$ are automatically declared winners,\footnote{However, the winners do not actually pay $p_h$, but rather a common price determined at the final round of the auction; this feature is necessary in order to provide any meaningful incentive compatibility guarantees.} while the agents below the lower price $p_{\ell}$ will have to quit the auction. Then, the mechanism may simply recurse on the agents that lie in the intermediate region. In this context, we consider the following encoding scheme: 

\begin{itemize}
    \item If $v_i > p_h$, then $i$ transmits a bit of $1$;
    \item If $v_i < p_{\ell}$, then $i$ transmits a bit of $0$;
    \item otherwise, $i$ may transmit some arbitrary 2-bit code.
\end{itemize}

Observe that the last condition ensures that the encoding is \emph{non-singular}. In contrast to our approach in the single-item auction, this communication pattern requires the transmission of a $2$-bit code from some agents; nonetheless, we will show that this overhead can be negligible, and in particular, the fraction of agents that reside between $p_h$ and $p_{\ell}$ can be made arbitrarily small. For simplicity in the exposition, here we will tacitly assume that the agents' valuations are pairwise distinct. The pseudocode for our mechanism is given in \ref{algorithm:multi_unit}.

\begin{algorithm}[htb]
\DontPrintSemicolon
\SetAlgoLined
\textbf{Input}: Set of agents $N$, number of items $m$ \;
\textbf{Output}: VCG outcome (Winners \& Payment)\;
Initialize the winners $W := \emptyset$ and the losers $L:=\emptyset$\;
$p_h := \textsc{EstimateUpperBound}(N, m)$ \;
$p_{\ell} := \textsc{EstimateLowerBound}(N, m)$ \;
Announce $p_{h}$ and $p_{\ell}$\;
Update the winners: $W := W \cup \{ i \in N \mid v_i > p_h \}$  \;
Update the losers: $L := L \cup \{ i \in N \mid v_i < p_{\ell} \}$\;
\eIf{$p_h = p_{\ell}$}{
\textbf{return} $(W, p_{h})$\;
}{
Set $m := m - |\{ i \in N : v_i > p_h \}| $\;
Set $N := N \setminus (W \cup L)$\;
\textbf{return} $\algomultiunit(N, m)$\;
}
\caption{$\algomultiunit(N, m)$}
\label{algorithm:multi_unit}
\end{algorithm}

The crux of the $\algomultiunit$ lies in the implementation of the subroutines at steps $4$ and $5$. This is is addressed in the following theorem:

\begin{theorem}
    \label{theorem:estimation}
Consider a set of agents $N = [n]$ and a number of units $m$. There exists a sampling algorithm such that for any $\epsilon > 0$ and any $\delta > 0$ satisfies the following:

\begin{itemize}
    \item It takes as input at most $4k \log(4k/\delta)/\epsilon^2$ bits.
    \item With probability at least $1 - \delta$ it returns prices $p_{h}$ and $p_{\ell}$, such that $p_{h}$ is between the $(m+1)$-ranked player and the $(m+1 + \lceil \epsilon n \rceil)$-ranked player, and $p_{\ell}$ is between the $(m+1 - \lceil \epsilon n \rceil)$-ranked player and the $(m+1)$-ranked player.
\end{itemize}
\end{theorem}

\begin{proof}
Consider a perfect binary tree of height $k$, such that each of the $2^k$ leaves corresponds to a point on the discretized valuation space, as illustrated in \Cref{fig:binary}. Our algorithm will essentially perform \emph{stochastic binary search} on this tree. To be precise, beginning from the root of the tree, we will estimate an additional bit of $p_{h}$ and $p_{\ell}$ in every level of the tree. Let us denote with $x_1, x_2, \dots, x_r$, with $x_i \in \{0, 1\}$, the predicted bits after $r$ levels. In the current level, we take a random sample $S$ of size $c$ \emph{with} replacement\footnote{We assume sampling with replacement to slightly simplify the analysis; our approach is also directly applicable when the sampling occurs \emph{without} replacement.}($S$ here is potentially a \emph{multiset}), and we query every agent $i \in S$ on whether $v_i \leq \overline{x_1 x_2 \dots x_r 0 1 1 \dots 1}$, where the threshold is expressed in binary representation. Let us denote with $X_{\mu}$ the random estimation derived from the sample, i.e.,
\begin{equation*}
    X_{\mu} = \frac{\sum_{i \in S}  \mathbbm{1} \{ v_i \leq \overline{x_1 x_2 \dots x_r 0 1 \dots 1} \}}{c}.
\end{equation*}

Recall that the output of the algorithm consists of two separate $k$-bit numbers $p_{h}$ and $p_{\ell}$. For convenience, a sample will be referred to as $\epsilon$-\emph{ambiguous} if $|X_{\mu} - \gamma| < \epsilon$, where $\gamma = m/n$ and $\epsilon > 0$ some parameter. Intuitively, whenever the sample is \emph{unambiguous} we can branch with very high confidence; that is, we predict a bit of $1$ if $X_{\mu} < \gamma$, and a bit of $0$ if $X_{\mu} > \gamma$. In contrast, in every $\epsilon$-ambiguous junction the ``high'' estimation---corresponding to $p_h$---will predict a bit of $1$, whilst the ``lower'' estimation---corresponding to $p_{\ell}$---will predict a bit of $0$. One should imagine that the two estimators initially coincide, until they separate when a ``close'' decision arises (see \Cref{fig:binary}). We claim that this algorithm will terminate with high probability with the desired bounds. For our analysis we will employ the following standard lemma:

\begin{lemma}[Chernoff-Hoeffding Bound]
    Let $\{X_1, X_2, \dots, X_c\}$ be a set of i.i.d. random variables with $X_i \sim \text{Bern}(p)$ and $X_{\mu} = (X_1 + X_2 + \dots + X_c)/c$; then,
    \begin{equation*}
        \Pr(|X_{\mu} - p| \geq \epsilon) \leq 2 e^{-2 \epsilon^2 c}.
    \end{equation*}
\end{lemma}

The main observation is that if for all samples $X_{\mu}$ has at most $\epsilon/2$ error, then the estimations $p_h$ and $p_{\ell}$ will be within the desired range in our claim. Let us denote with $p_e$ the probability that for a single estimate and after $k$ levels there exists a sample with more than $\epsilon/2$ error; the union bound implies that $p_e \leq 2k e^{-\epsilon^2 c/2}$. Thus, for any $\epsilon > 0$ and $\delta > 0$, $p_e \leq \delta$ for $c \geq 2 \log(2k/\delta)/\epsilon^2$. Furthermore, by the union bound we obtain that the probability of error of either of the two estimates with input at most $2 c k$ bits is at most $2 \delta$; rescaling $\delta := \delta/2$ concludes the proof.

\end{proof}

\begin{figure}[!ht]
    \centering
    \includegraphics[scale=0.5]{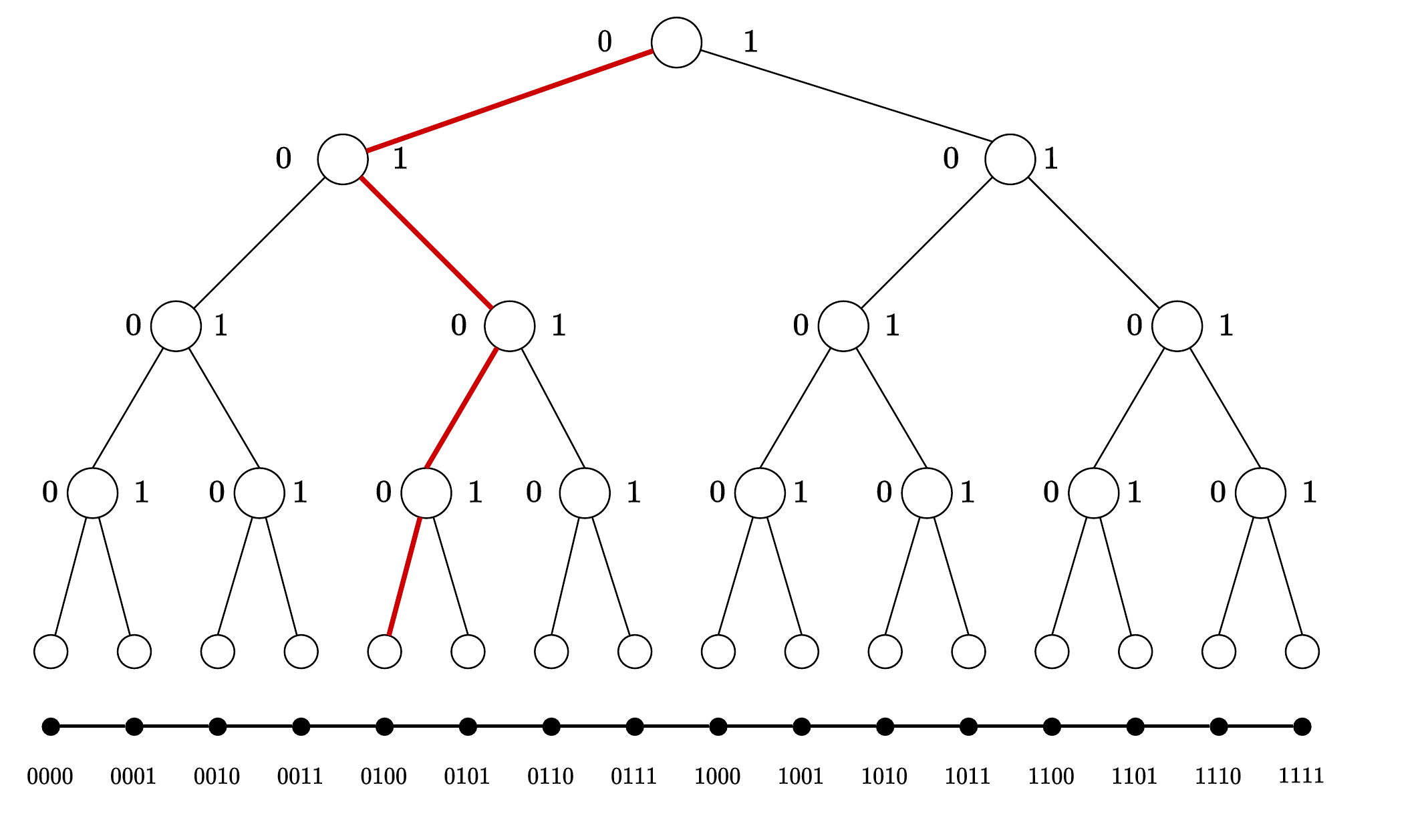}
    \caption{The binary-tree representation of the valuation space. The red lines correspond to potential branching of the two estimates.}
    \label{fig:binary}
\end{figure}

Consequently, the algorithm described in \Cref{theorem:estimation} will be employed to implement lines $4$ and $5$ in our $\algomultiunit$. In addition, notice that we can recognize whenever the estimated prices $p_{h}$ and $p_{\ell}$ are incongruous---in the sense that either the winners are more than the available items, or that the remaining agents are less than the available items, in which case we simply repeat the estimation. Thus, we obtain the following properties:

\begin{proposition}
    The $\algomultiunit$ is ex-post incentive compatible.
\end{proposition}

\begin{theorem}[Multi-Unit Auction with Optimal Communication]
    \label{theorem:multi-unit-optimal}
Consider some set of unit demand agents $N = [n]$, and $m$ identical units. Moreover, denote with $t(n; k)$ the expected communication complexity of $\algomultiunit(N, m)$, with steps $4$ and $5$ implemented through the algorithm of \Cref{theorem:estimation}. If $k = \mathcal{O}(n^{1 - \ell})$ for some $\ell > 0$, then for any $\epsilon > 0$,
\begin{equation*}
    t(n; k) \lesssim (1 + \epsilon) n.
\end{equation*}
\end{theorem}

\begin{proof}
\Cref{theorem:estimation} implies that for any $\epsilon > 0$ and $\delta > 0$,

\begin{equation*}
    \label{eq:rec}
    t(n; k) \leq (1 - \delta) ((1 + 2\epsilon) n + t(2\epsilon n; k)) + \delta (2n + t(n; k)) + 4k \frac{\log(4k/\delta)}{\epsilon^2},
\end{equation*}
where the first term corresponds to the induced communication when the sampling algorithm of \Cref{theorem:estimation} succeeds, the second term is the worst-case communication whenever the sampling algorithm fails to return prices within the desired range, while the last term is the communication of the sampling algorithm. Thus, solving the induced recursion and using that $k = \mathcal{O}(n^{1 - \ell})$ concludes the proof.
\end{proof}

\section{Concluding Remarks}

The main contribution of this work is twofold. First, we studied the performance of sampling approximations in facility location games. Our main result in that regime was to show that for any $\epsilon > 0$, we can obtain in expectation a $1 + \epsilon$ approximation w.r.t. the optimal social cost in the metric space $(\mathbb{R}^d, \|\cdot\|_1)$, with only a sample of size $c(\epsilon) = \Theta(1/\epsilon^2)$. Moreover, we considered a series of exemplar environments from auction theory; for every instance, we proposed a mechanism that elicits $1 + \epsilon$ bits on average per bidder, for any $\epsilon > 0$, asymptotically matching the communication lower bound. While our main focus in this work was on social welfare guarantees, an interesting avenue for future research would be to consider other natural objectives in auctions, such as revenue maximization.


\section*{Acknowledgments}

This work was supported by the Hellenic Foundation for Research and Innovation (H.F.R.I.) under the “First Call for H.F.R.I. Research Projects to support Faculty members and Researchers and the procurement of high-cost research equipment grant”, project BALSAM, HFRI-FM17-1424.

\bibliography{paper}

\clearpage

\appendix

\section{Sampling Approximation of the Plurality Rule}
\label{appendix:plurality}
Approximating the plurality rule is quite folklore in social choice; e.g., see \cite{10.5555/2772879.2773334,CANETTI199517}. More recently, Bhattacharyya and Dey \cite{10.5555/2772879.2773334} analyzed the sample complexity of determining the outcome in several common voting rules under a $\gamma$-\emph{margin} condition; that is, they assumed that the minimum number of votes that need to be modified in order to alter the winner is at least $\gamma \cdot n$.  In fact, a standard result by Canetti et al. \cite{CANETTI199517} establishes that at least $\Omega (\log(1/\delta)/\gamma^2)$ samples are required in order to determine the winner in the plurality rule with probability at least $1 - \delta$, even with $2$ candidates. This lower bound might appear rather unsatisfactory; for one thing, the designer does not typically have any prior information on the margin $\gamma$. More importantly, in many practical scenarios we expect the margin to be negligible, leading to a substantial overhead in the sample complexity. 

The purpose of this section is to show that these obstacles could be circumvented once we take a \emph{utilitarian} approach. Indeed, instead of endeavoring to determine the winner in the election with high probability, we are interested in approximating the social welfare. More precisely, assume that $n$ agents have to select among $m$ alternatives or candidates. We let $u_{i, j}$ represent the score that agent $i \in [n]$ assigns to candidate $j \in [m]$. In this way, the social welfare of an outcome $j \in [m]$ is defined as 

\begin{equation*}
    \SW(j) = \sum_{i=1}^n u_{i, j}.
\end{equation*}

Notice, however, that the social welfare approximation problem through the plurality rule is hopeless under arbitrary valuations, in light of obvious information-theoretic barriers. (the framework of \emph{distortion} introduced by Procaccia and Rosenschein \cite{10.1007/11839354_23} quantifies exactly these limitations.) For this reason, and for the sake of simplicity, we are considering the social welfare approximation problem in a very simplistic setting.

\begin{definition}
A voter $i$ is said to be single-minded if $u_{i, r} = 1$ for some candidate $r \in [m]$, and $u_{i, j} = 0, \forall j \neq r$.
\end{definition}

In fact, in the simple setting where all agents are single-minded it is easy to see that the plurality rule is actually strategy-proof. Recall that in the plurality rule every agent $i$ votes for a \emph{single} candidate $j \in [m]$, i.e. agent $i$ broadcasts an $m$-tuple $(0, 0, \dots, 1, 0)$, assigning $1$ to the position that corresponds to her preferred candidate. We are now ready to analyze the sampling approximation of the plurality rule.

\begin{algorithm}
\label{algorithm:approx_plurality}
\DontPrintSemicolon
\SetAlgoLined
\textbf{Input}: Set of agents $N$, set of candidates $[m]$, accuracy parameter $\epsilon > 0$, confidence parameter $\delta > 0$\;
\textbf{Output}: Candidate $j \in [m]$ such that $\SW(j) \geq (1 - \epsilon) \SW(j^*)$ with probability at least $1 - \delta$, where $j^*$ represents the optimal candidate \;
Set $c = 2m^2 \log(2m/\delta)/\epsilon^2$, the size of the sample\;
Let $S$ be a random sample\footnotemark of $c$ agents from $N$\;
\textbf{return} $\textsc{Plurality}(S)$\;
\caption{$\algoplurality(N, [m], \epsilon, \delta)$}
\end{algorithm}
\footnotetext{Here we assume that we sample with replacement.}

\begin{theorem}
\label{theorem:plurality}
Consider a set of single-minded agents $N$, and any number of candidates $m$. For any $\epsilon > 0$ and any $\delta > 0$, $\algoplurality(N, [m], \epsilon, \delta)$ yields with probability at least $1 - \delta$ an approximation ratio of $1 - \epsilon$ w.r.t. the optimal social welfare of the full information $\textsc{Plurality}$, for any $c \geq c_0(m, \epsilon, \delta)$, where 
\begin{equation*}
    c_0(m, \epsilon, \delta) = \frac{2 m^2 \log(2m/\delta)}{\epsilon^2}.
\end{equation*}
\end{theorem}

The proof of this theorem is simple and proceeds with a standard Chernoff bound argument; for completeness, we provide a detailed proof.

\begin{proof}[Proof of \Cref{theorem:plurality}]
First, let us denote with $s_j = \left(\sum_{i=1}^n u_{i, j}\right)/n$ the score of the $j^{\text{th}}$ candidate in the full information plurality rule. Consider a sample of size $c$, and let $(X_{i, 1}, \dots, X_{i, m})$ represent the voting profile of the agent that was selected in the $i^{\text{th}}$ iteration of the sampling process. It follows that $X_{i, j} \sim \text{Bern}(s_j), \forall j \in [m]$. Moreover, recall that we consider sampling with replacement, so that the set of random variables $\{X_{1, j}, \dots, X_{c, j}\}$ is pairwise independent for any $j \in [m]$; our results are also applicable when the sampling occurs without replacement given that the correlation between the random variables is negligible -- for sufficient large $n$, although we do not formalize this here. 

As a result, if we denote with $\widehat{s_j} = \left(\sum_{i=1}^c X_{i, j}\right)/c$, a standard Chernoff bound argument implies that $\forall \epsilon' > 0, \forall \delta' > 0$, and $c \geq \log(2/\delta')/(2(\epsilon')^2)$, $|\widehat{s_j} - s_j| \leq \epsilon'$ with probability at least $1 - \delta'$. By the union bound, we obtain that $|\widehat{s_j} - s_j| \leq \epsilon'$ for all $j \in [m]$ and with probability at least $1 - m\delta'$. Thus, we let $\delta' = \delta/m$ for some arbitrary $\delta > 0$. We also let $s^* = s_{j^*}$ to be the score of the most preferred candidate $j^*$ -- among the entire population $N$. By definition, the optimal social welfare is $\SW^* = \SW(j^*) = \sum_{i=1}^n u_{i, j^*} = n s^*$. If $r = \argmax_j \widehat{s_j}$ is the random output of $\algoplurality(N, [m], \epsilon', \delta)$, it follows that $s_{r} \geq s^* - 2\epsilon'$ with probability at least $1 - \delta$. Thus, we obtain that 

\begin{align*}
    \SW(r) = \sum_{i=1}^n u_{i, r} = n s_r &\geq n (s^* - 2\epsilon') \\
    &= \SW^* - 2\epsilon' \frac{\SW^*}{s^*} \\
    &\geq \SW^* (1 - 2\epsilon' m),
\end{align*}
where in the final inequality we used that $s^* \geq 1/m$. Finally, setting $\epsilon' = \epsilon/(2m)$ for an arbitrary $\epsilon > 0$ concludes the proof.
\end{proof}

\clearpage

\section{Numerical Bounds on the Vandermonde Distribution}
\label{appendix:numerical}

The purpose of this section is to supplement our asymptotic analysis with numerical bounds on the convergence and the concentration of the Vandermonde distribution (recall \eqref{eq:pmf}). 

\subsection{Convergence of the Vandermonde Distribution}

First, we provide experiments that illustrate the rapid convergence of the Vandermonde distribution, both qualitatively (\Cref{fig:convergence}) and in quantitative form (\Cref{table:convergence}).

\begin{figure}[!ht]
    \centering
    \includegraphics[scale=0.4]{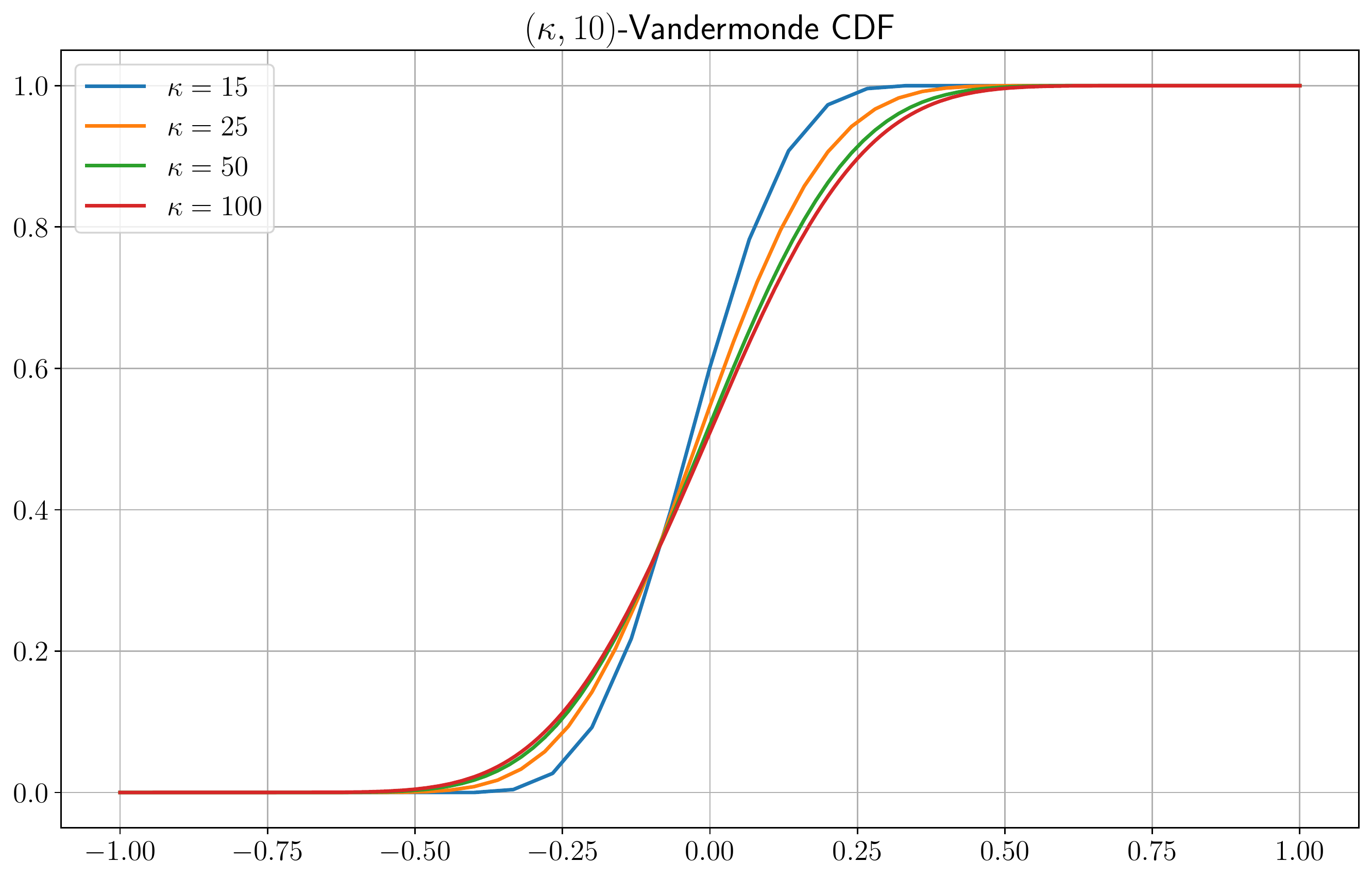}
    \caption{The cumulative distribution function (CDF) of the $(\kappa, 10)$-Vandermonde distribution for $\kappa \in \{15, 25, 50, 100\}$.  }
    \label{fig:convergence}
\end{figure}

\begin{table}[!ht]
\centering
 \begin{tabular}{c||c|c|c|c|c} 
 $ \delta(\mu_{\rho}^{\kappa}, \mu_{\rho}^{\infty}) $ & $\kappa=100$ & $\kappa=500$ & $\kappa=1000$ & $\kappa = 2000$ & $\kappa = 5000$ \\ [0.5ex] 
 \hline\hline
 $\rho=10$ & $0.0472$ & $0.0090$ & $0.0045$ & $0.0022$ & $0.0009$ \\ \hline
 $\rho=20$ & $0.1041$ & $0.0190$ & $0.0094$ & $0.0047$ & $0.0019$ \\
 \hline
 $\rho=30$ & $0.1685$ & $0.0292$ & $0.0144$ & $0.0071$ & $0.0028$ \\
 \hline
 $\rho=40$ & $0.2422$ & $0.0396$ & $0.194$ & $0.0096$ & $0.0038$ \\
 \hline
 $\rho=50$ & $0.3286$ & $0.0502$ & $0.0245$ & $0.0121$ & $0.0048$ \\ 
 [1ex] 
\end{tabular}
\caption{The convergence of the Vandermonde distribution in terms of \emph{total variation distance}. Here $\mu_{\rho}^{\kappa}$ represents the probability measure of the $(\kappa, \rho)$-Vandermonde distribution, and $\delta(\mu, \mu')$ the total variation distance between measures $\mu$ and $\mu'$. We remark that for our experiments we \emph{discretized} the measure $\mu_{\rho}^{\infty}$.}
\label{table:convergence}
\end{table}

\subsection{Concentration of the Vandermonde Distribution}

Next, we illustrate how the size of the sample, or equivalently parameter $\rho$, affects the concentration of the Vandermonde distribution. 

\begin{figure}[!ht]
    \centering
    \includegraphics[scale=0.4]{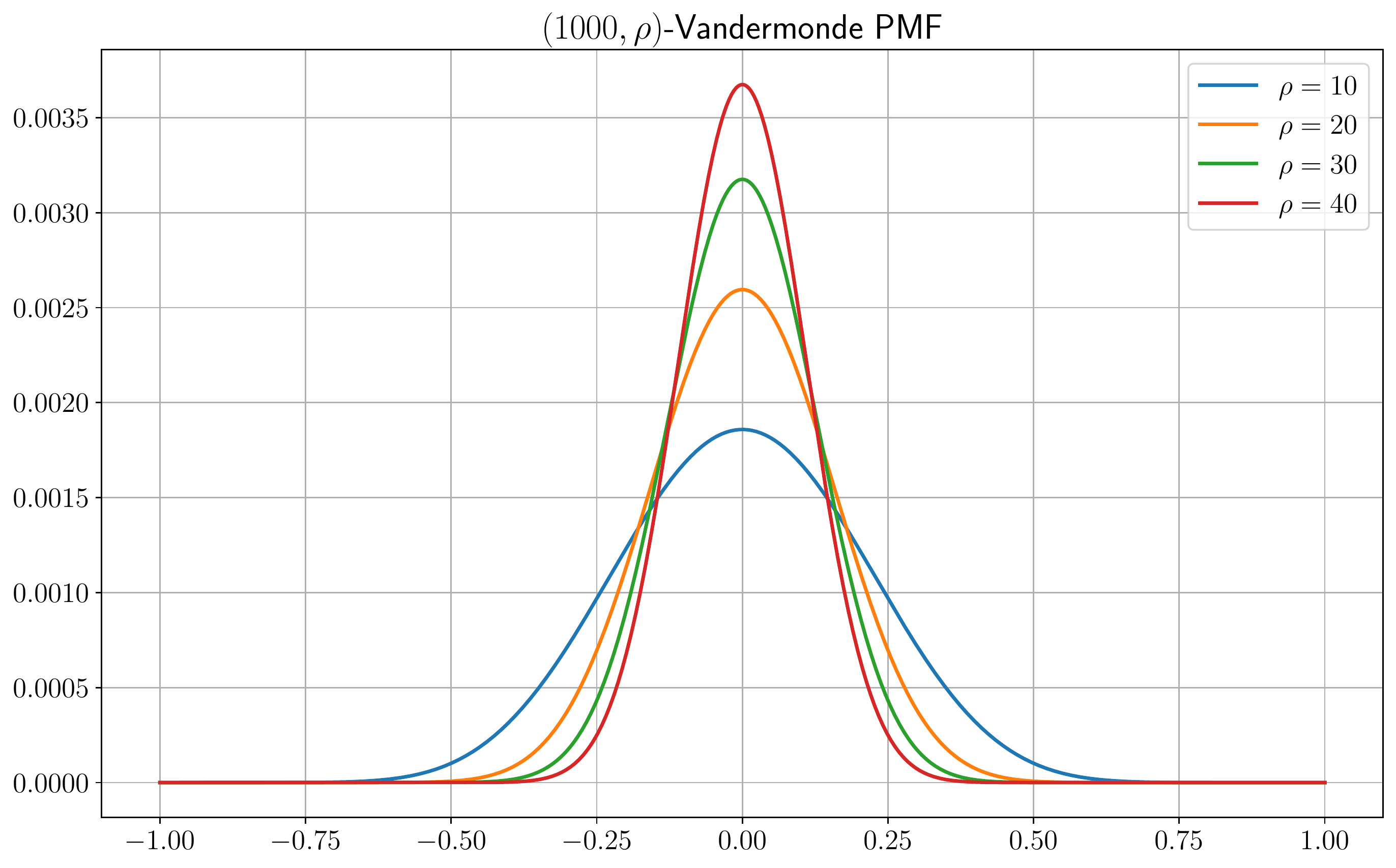}
    \caption{The probability mass function of the $(1000, \rho)$-Vandermonde distribution for $\rho \in \{10, 20, 30, 40\}$. Naturally, as the size of the sample increases the probability mass accumulates closer to $0$, which corresponds to the median of the entire instance.}
    \label{fig:vandermonde_1}
\end{figure}

\begin{table}[!ht]
\centering
 \begin{tabular}{c||c|c|c|c|c} 
 $\E[|X_r|]$ & $\kappa=100$ & $\kappa=500$ & $\kappa=1000$ & $\kappa = 2000$ & $\kappa = 5000$ \\ [0.5ex] 
 \hline\hline
 $\rho=10$ & $0.1603$ & $0.1667$ & $0.1674$ & $0.1678$ & $0.1680$ \\ \hline
 $\rho=20$ & $0.1100$ & $0.1200$ & $0.1212$ & $0.1218$ & $0.1222$ \\
 \hline
 $\rho=30$ & $0.0848$ & $0.0979$ & $0.0994$ & $0.1002$ & $0.1006$ \\
 \hline
 $\rho=40$ & $0.0683$ & $0.0843$ & $0.0861$ & $0.0870$ & $0.0875$ \\
 \hline
 $\rho=50$ & $0.0559$ & $0.0748$ & $0.0769$ & $0.0778$ & $0.0784$ \\ 
 [1ex] 
\end{tabular}
\caption{The expectation of $|X_r|$, where $X_r$ is drawn from a $(\kappa, \rho)$-Vandermonde distribution, for different values of $\kappa$ and $\rho$. Notice that as $\kappa$ increases, the expectation $\E[|X_r|]$ converges rapidly. }
\label{table:median}
\end{table}

\end{document}